\documentclass[12pt,reqno]{amsart}

\usepackage[active]{srcltx}

\usepackage{latexsym}
\usepackage{graphicx}

\usepackage{graphicx,amsfonts,amssymb,amsmath,amsthm}
\usepackage{fullpage}
\usepackage[all,arc]{xy}
\usepackage[dvipsnames]{xcolor}
\usepackage{hyperref}
\usepackage{mathrsfs,mathtools}
\usepackage{caption}
\usepackage{enumitem}
\usepackage{tikz-cd,float}
\usepackage{pdflscape,bm,subfig}
\usetikzlibrary{decorations.pathreplacing,calligraphy}

\usepackage{color}

\renewcommand{\Re}{\operatorname{Re}}
\renewcommand{\Im}{\operatorname{Im}}
\newtheorem{theo}{{\sc Theorem}}[section]

\newtheorem{lem}[theo]{{\sc Lemma}}
\newtheorem{prop}[theo]{{\sc Proposition}}

\newenvironment{rem}{\medskip\noindent{\it Remark:\/} }{\medskip}
\allowdisplaybreaks

\definecolor{BabyBlue}{rgb}{.5373,.8118,.9412}
\definecolor{BabyRed}{rgb}{1.0, .4667, .4745}
\makeatletter
\def\blfootnote{\gdef\@thefnmark{}\@footnotetext}
\makeatother

\title{Scaling Asymptotics of Wigner Distributions of Harmonic Oscillator Orbital Coherent States}

\author{Nicholas Lohr}

\date{\today}

\begin{document}
\blfootnote{The author was partially supported by NSF RTG grant DMS-1502632.}

\maketitle
\begin{abstract}
The main result of this article gives scaling asymptotics of the Wigner distributions $W_{\varphi_N^{\gamma},\varphi_N^{\gamma}}$ of isotropic harmonic oscillator \textit{orbital coherent states} $\varphi_N^{\gamma}$ concentrating along Hamiltonian orbits $\gamma$ in shrinking tubes around $\gamma$ in phase space. In particular, these Wigner distributions exhibit a \textit{hybrid} semi-classical scaling. That is, simultaneously, we have an \textit{Airy scaling} when the tube has radius $N^{-2/3}$ normal to the energy surface $\Sigma_E$, and a \textit{Gaussian scaling} when the tube has radius $N^{-1/2}$ tangent to $\Sigma_E$.
\end{abstract}
\tableofcontents
\section{Introduction}
Coherent states centered at points in phase space were introduced by Schr\"odinger \cite{S26} in 1926 as minimal uncertainty states. In 1963, Glauber \cite{G63}
extended coherent states to quantum electrodynamics. This paper studies \textit{orbital coherent states} centered on periodic classical Hamiltonian orbits. Orbital coherent states were first studied in 1989 for the Hydrogen atom in the papers of Gay et al. \cite{GDB89} and Nauenberg \cite{N89}, and more recently in papers by Klauder \cite{K96} and Villegas-Blas et al. \cite{VB96,TVB97}. In 1992, De Bi\`evre \cite{De92} studied orbital coherent states \eqref{eq:coherent} for the
two-dimensional isotropic harmonic oscillator:
\begin{equation}\label{eq:harmos}
\widehat{H}_{\hbar} \coloneqq -\frac{\hbar^2}{2}\Delta_{\mathbb{R}^2}+\frac{|x|^2}{2}.
\end{equation}
For $\gamma_0(t)\coloneqq \sqrt{E}(\cos t, 
\sin t,-\sin t,\cos t)$ in $T^*\mathbb{R}^2$ with fixed energy $E=\hbar (N+1)$, the $L^2(\mathbb{R}^2)$-normalized orbital coherent state associated to $\gamma_0$ is
\begin{equation}\label{eq:coherent}
 \varphi_N^{\gamma_0}(x_1,x_2)=\frac{(N+1)^{\frac{N+1}{2}}}{E^{\frac{N+1}{2}}\sqrt{\pi
    N!}}(x_1+ix_2)^Ne^{-\tfrac{(N+1)}{2E}|x|^2}.
\end{equation}
From this initial coherent state $\varphi_{N}^{\gamma_0}$, we define $\varphi_N^{\gamma}$ for a general Hamiltonian orbit $\gamma$ with the metaplectic representation (see \eqref{eq:genphi} of Section \ref{s:statements}).
This article concerns the Wigner distributions of these orbital coherent states.  We recall that in dimension 2, the Wigner distribution $W_{\varphi_N^{\gamma},\varphi_N^{\gamma}}:T^*\mathbb{R}^2 \to \mathbb{R}$ of $\varphi_N^{\gamma}$ is
defined by
\begin{align*}
W_{\varphi_N^{\gamma},\varphi_N^{\gamma}}(x,\xi) \coloneqq &\
 \frac{1}{4\pi^2\hbar^2}\int_{\mathbb{R}^2}\varphi_N^{\gamma}(x+\tfrac{v}{2})\overline{\varphi_N^{\gamma}(x-\tfrac{v}{2})}e^{-\frac{i}{\hbar
                                                             }\langle
                                         v ,
                                                             \xi
                                                             \rangle}\
                                                             dv
\end{align*}
(see \cite{W32}).
It is well-known that for fixed energy, the Wigner distributions tend in the weak
sense to a delta function along the orbit $\gamma$ as $N \to \infty$, i.e.
\begin{equation}\label{eq:ah}
W_{\varphi_N^{\gamma},\varphi_N^{\gamma}} \rightharpoonup \delta_{\gamma}
\end{equation}
(see \cite{De92,De93}). For completeness, we prove \eqref{eq:ah} in Proposition \ref{t:delta}. Our purpose is to study the
fine structure of this limit in small tubes around $\gamma$. The main result of this article gives Airy and Gaussian scaling asymptotics
of the Wigner distribution
$W_{\varphi_N^{\gamma},\varphi_N^{\gamma}}$ of
$\varphi_N^{\gamma} \in L^2(\mathbb{R}^2)$ (see Theorem
\ref{t:scale}). In this article, we prove all of the results in dimension 2 for notational convenience. These results generalize to dimension $d$ and the same arguments hold, as explained in Section \ref{s:generalize}.
\subsection{Statements of Results}\label{s:statements}
The orbital coherent states \eqref{eq:coherent} are eigenstates of \eqref{eq:harmos},
\begin{equation*}
\widehat{H}_{\hbar}\varphi_N^{\gamma_0}=E \varphi_N^{\gamma_0}
\end{equation*}
where $E=\hbar (N+1)$. The non-constant Hamiltonian orbits of the classical Hamiltonian $H(x,\xi)=\frac{1}{2}|\xi|^2+\frac{1}{2}|x|^2$ are great circles on the energy surface $\Sigma_E:=\{(x,\xi) : H(x,\xi)=E\}$. Modding out by the points on the same orbit, we see that the space of non-constant orbits is isomorphic to $S^3/S^1=\mathbb{C}\text{P}^1$. Since $\operatorname{SU}(2)$ acts transitively on $\mathbb{C}\text{P}^1$ by rotation, for any non-constant Hamiltonian orbit $\gamma$ there exists a $U \in \operatorname{SU}(2)$ such that $U \cdot \gamma_0=\gamma$. We define (up to a phase factor) the coherent state centered at the Hamiltonian orbit $\gamma$ by
\begin{equation}\label{eq:genphi}
\varphi_N^{\gamma} \coloneqq \mu(U)\varphi_N^{\gamma_0}
\end{equation}
where $\mu$ is the metaplectic representation on $L^2(\mathbb{R}^2)$
(see Chapter 4 of \cite{F89} for details on the metaplectic representation).
\par The first result of this article is a proof of \eqref{eq:ah}.
\begin{prop}\label{t:delta} Let $a:T^*\mathbb{R}^2 \to \mathbb{R}$ be a smooth function with exponential decay. Then in the limit $N\to \infty$ where $E=\hbar (N+1)$, we have
\begin{equation}\label{eq:me}
\int_{T^*\mathbb{R}^2}a(x,\xi)W_{\varphi_N^{\gamma},\varphi_N^{\gamma}}(x,\xi)\ dxd\xi \xrightarrow{N \to \infty} \frac{1}{2\pi}\int_0^{2\pi}a\big(\gamma(t)\big)dt.
\end{equation}
\end{prop}
It was proved in \cite{De92} that given a smooth function $a:T^*\mathbb{R}^2 \to \mathbb{R}$ with exponential decay, in the limit $N\to \infty$ where $E=\hbar (N+1)$, we have
 \begin{equation}\label{eq:deb}
\langle
 \varphi_N^{\gamma},\operatorname{Op}_N^{aw}(a)[\varphi_N^{\gamma}]\rangle \xrightarrow{N \to \infty} \frac{1}{2\pi}\int_0^{2\pi}a\big(\gamma(t)\big)dt,
\end{equation}
 where $\operatorname{Op}_N^{aw}(a)$ is the anti-Wick quantization of $a$. Proposition \ref{t:delta} proves the analogous result for Weyl quantization, based on the well-known formula
\begin{equation}\label{eq:wig1}
\langle
 \varphi_N^{\gamma},\operatorname{Op}_N^{w}(a)[\varphi_N^{\gamma}]\rangle = \int_{T^*\mathbb{R}^2}a(x,\xi)W_{\varphi_N^{\gamma},\varphi_N^{\gamma}}(x,\xi)dxd\xi.
\end{equation}
The next result gives the pointwise asymptotics of
$W_{\varphi_N^{\gamma},\varphi_N^{\gamma}}(x,\xi)$ for $(x,\xi)$ on and off the orbit $\gamma$ (see Figure \ref{fig:pointwise} for a picture).
\begin{theo}\label{t:main2}
As $N \to \infty$, we have the following pointwise asymptotics:
\begin{itemize}
\item[(1)] If $(x,\xi)$ lies on the orbit, then we have
$$\frac{1}{N^{5/3}}W_{\varphi_N^{\gamma},\varphi_N^{\gamma}}(x,\xi)=\frac{1}{2^{1/3}\pi^2E^2}\operatorname{Ai}(0)+O(N^{-2/3})$$
where $\operatorname{Ai}$ is the Airy function defined in the Section \ref{s:airy}. In terms of the Gamma function, $\frac{1}{2^{1/3}\pi^2E^2}\operatorname{Ai}(0)=\frac{\Gamma(1/3)12^{1/3}}{4\pi^3E^2\sqrt{3}}$
\item[(2)] If $(x,\xi)$ lies on the orbit, then for $0<t< 1$ we have
$$\frac{1}{N^{5/3}}W_{\varphi_N^{\gamma},\varphi_N^{\gamma}}\big(t(x,\xi)\big)=O(N^{-1/6}).$$
\item[(3)] If $(x,\xi)\neq 0$ is not in the above case, we have
$$W_{\varphi_N^{\gamma},\varphi_N^{\gamma}}(x,\xi) =
O(N^{-\infty}).$$
\item[(4)] If $(x,\xi)=0$, then the Wigner distribution is explicitly $$W_{\varphi_N^{\gamma},\varphi_N^{\gamma}}(x,\xi)=\frac{(-1)^N}{\pi^2E^2}(N+1)^2.$$
\end{itemize}
\end{theo}
\begin{figure}[!htbp]
    \centering
    \subfloat[\centering The Wigner distribution when $N=10$ and $E=1$]{{\includegraphics[width=8cm]{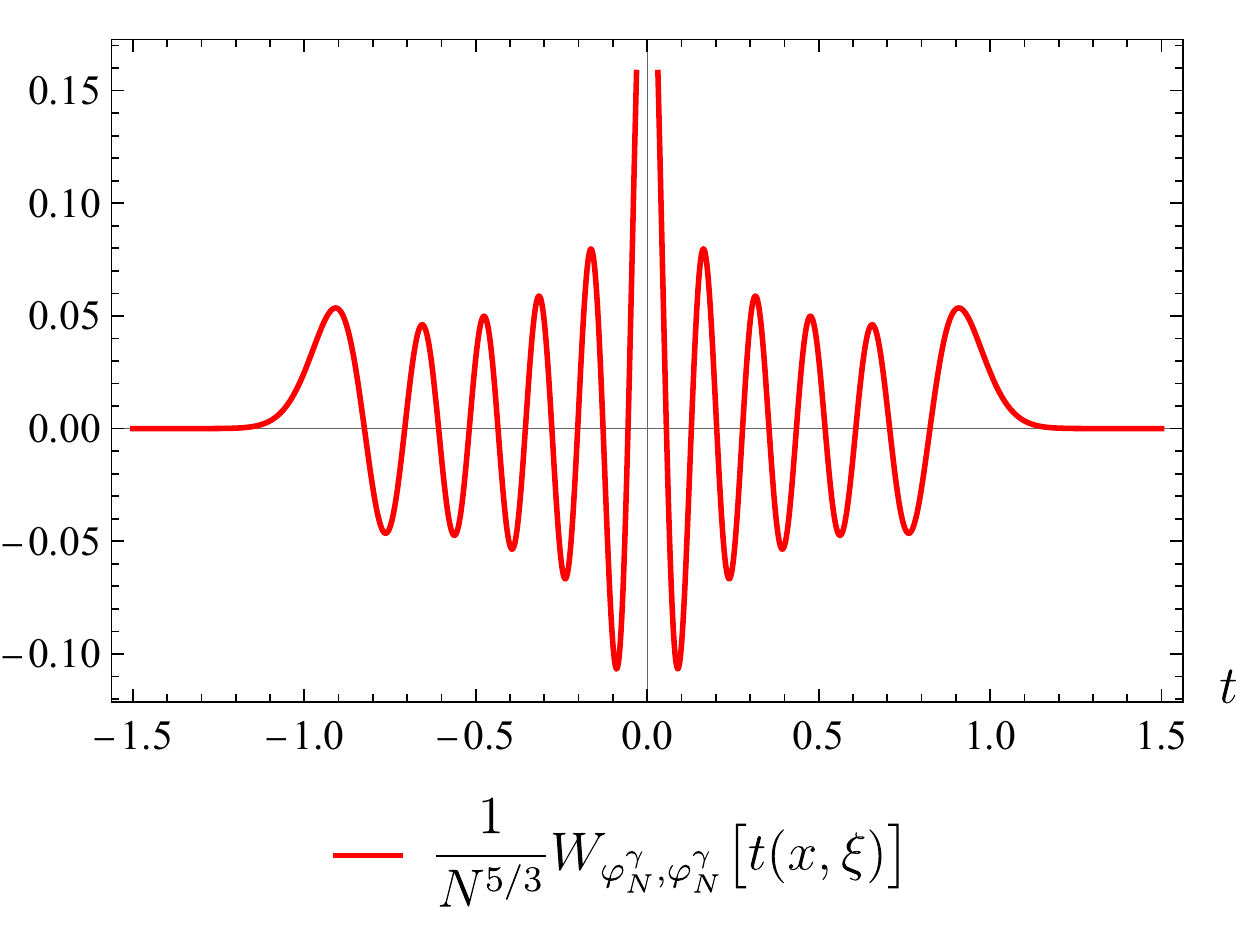} }}%
    \subfloat[\centering The Wigner distribution on the plane containing $\gamma$ when $N=10$ and $E=1$]{{\includegraphics[width=8cm]{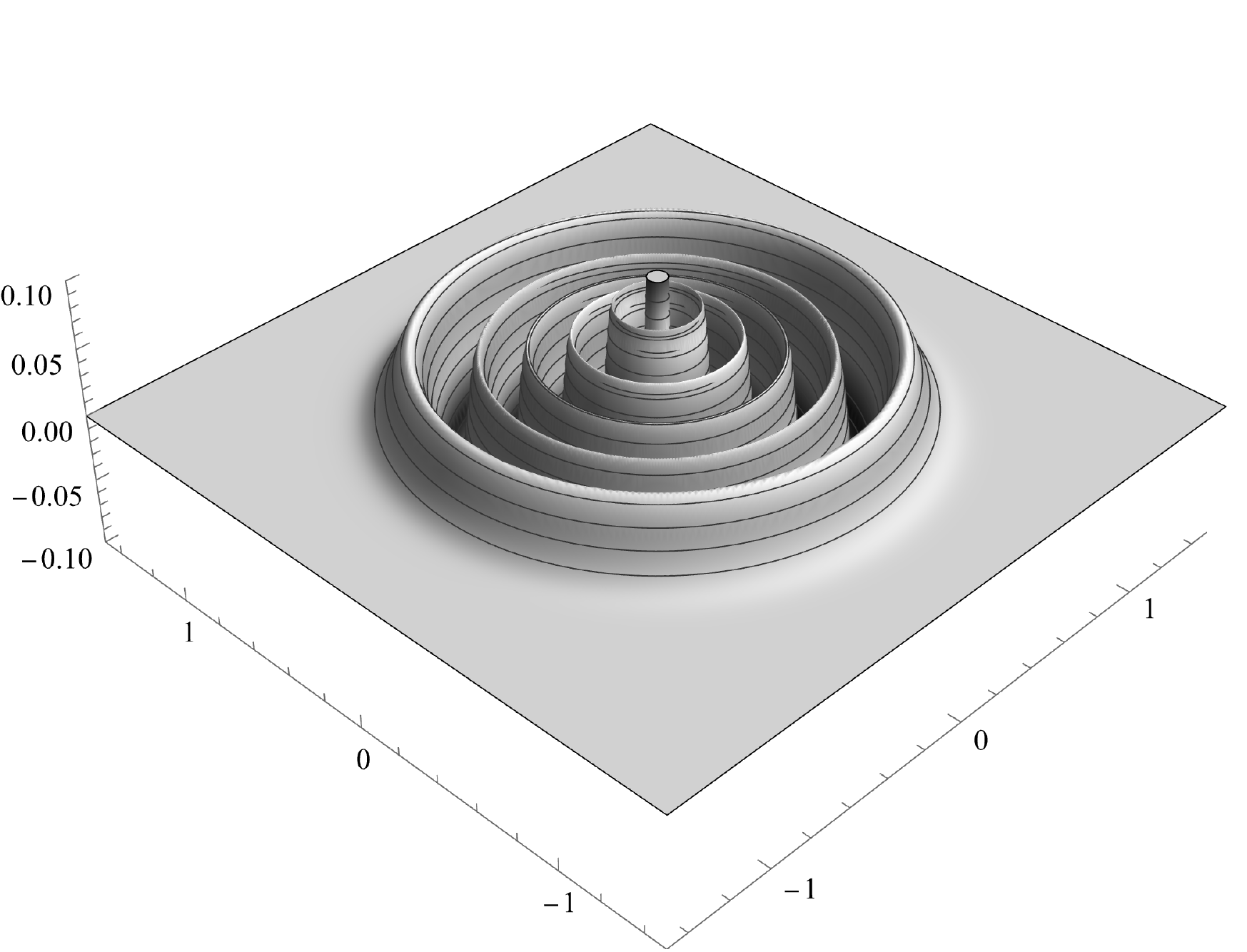} }}
    \caption{}\label{fig:pointwise}
    \end{figure}
Our main result concerns the scaling asymptotics between the power
growth on the orbit and the exponential decay off of it (see Figures \ref{fig:1} and \ref{fig:2} for pictures).
\begin{theo}\label{t:scale} Suppose
$(x,\xi)$ lies on $\gamma$, and let $(x,\xi),v_1,v_2$ be an orthogonal basis (with each vector of norm $\sqrt{2E}$) for the normal space, $N_{(x,\xi)}\gamma$, of $\gamma$ at $(x,\xi)$. For $(u,w_1,w_2) \in \mathbb{R}^3$
 in a sufficiently small neighborhood of $(0,0,0)$, we have the uniform asymptotic expansion
 \begin{align*}
   \frac{1}{N^{5/3}}W_{\varphi_N^{\gamma},\varphi_N^{\gamma}}\big[(1+u)(x,\xi)&+w_1v_1+w_2v_2\big]\\
                                                                                  &=
                                                          \frac{ e^{-2(N+1)(w_1^2+w_2^2)}}{4\pi^3E^{2}}                                                 
                                                         \mu_{00}(u)\operatorname{Ai}(a(u)N^{2/3})+\operatorname{Ai}'(a(u)N^{2/3})\cdot O(N^{-2/3}) 
 \end{align*}
 where
 $$a(u)^{3/2}= \begin{cases}
      3(u+1)i\sqrt{|u|(u+2)}-3i \operatorname{arccos}(1+u)&\text{ if }u<0\\
      3(u+1)\sqrt{u(u+2)}-3\operatorname{arccosh}(1+u)  &\text{ if }u\geq 0
      \end{cases}$$
and $\mu_{00}(u)$ is a smooth function (explicitly given in Section \ref{s:mu}) with $\mu_{00}(0)=2^{5/3}\pi$.
In particular, if we rescale $u \to \frac{u}{(2N)^{\frac{2}{3}}}$, $w_1 \to \frac{w_1}{(2N)^{\frac{1}{2}}}$, and $w_2 \to \frac{w_2}{(2N)^{\frac{1}{2}}}$, we have the pointwise limit
$$\lim_{N \to \infty}\frac{1}{N^{5/3}}W_{\varphi_N^{\gamma},\varphi_N^{\gamma}}\big[(1+\tfrac{u}{(2N)^{2/3}})(x,\xi)+\tfrac{w_1}{(2N)^{1/2}}v_1+\tfrac{w_2}{(2N)^{1/2}}v_2\big]= \frac{e^{-(w_1^2+w_2^2)}}{2^{1/3}\pi^2E^2}\operatorname{Ai}(2u).$$
\end{theo}
\begin{figure}[!htbp]
    \centering
    \subfloat[\centering The Wigner distribution and the scaling when $N=50$ and $E=1$]{{\includegraphics[width=8cm]{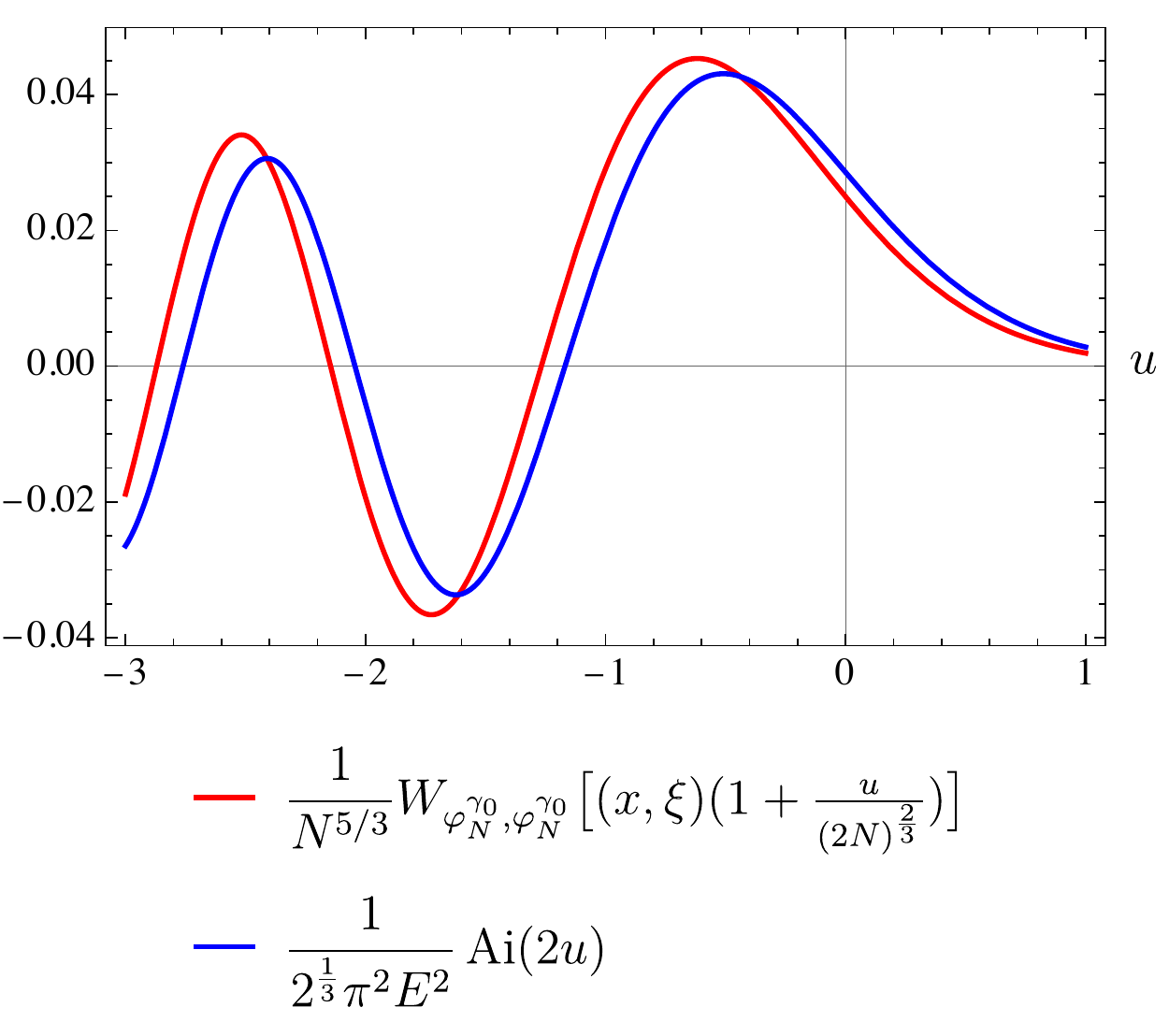} }}%
    \subfloat[\centering The Wigner distribution and the scaling when $N=200$ and $E=1$]{{\includegraphics[width=8cm]{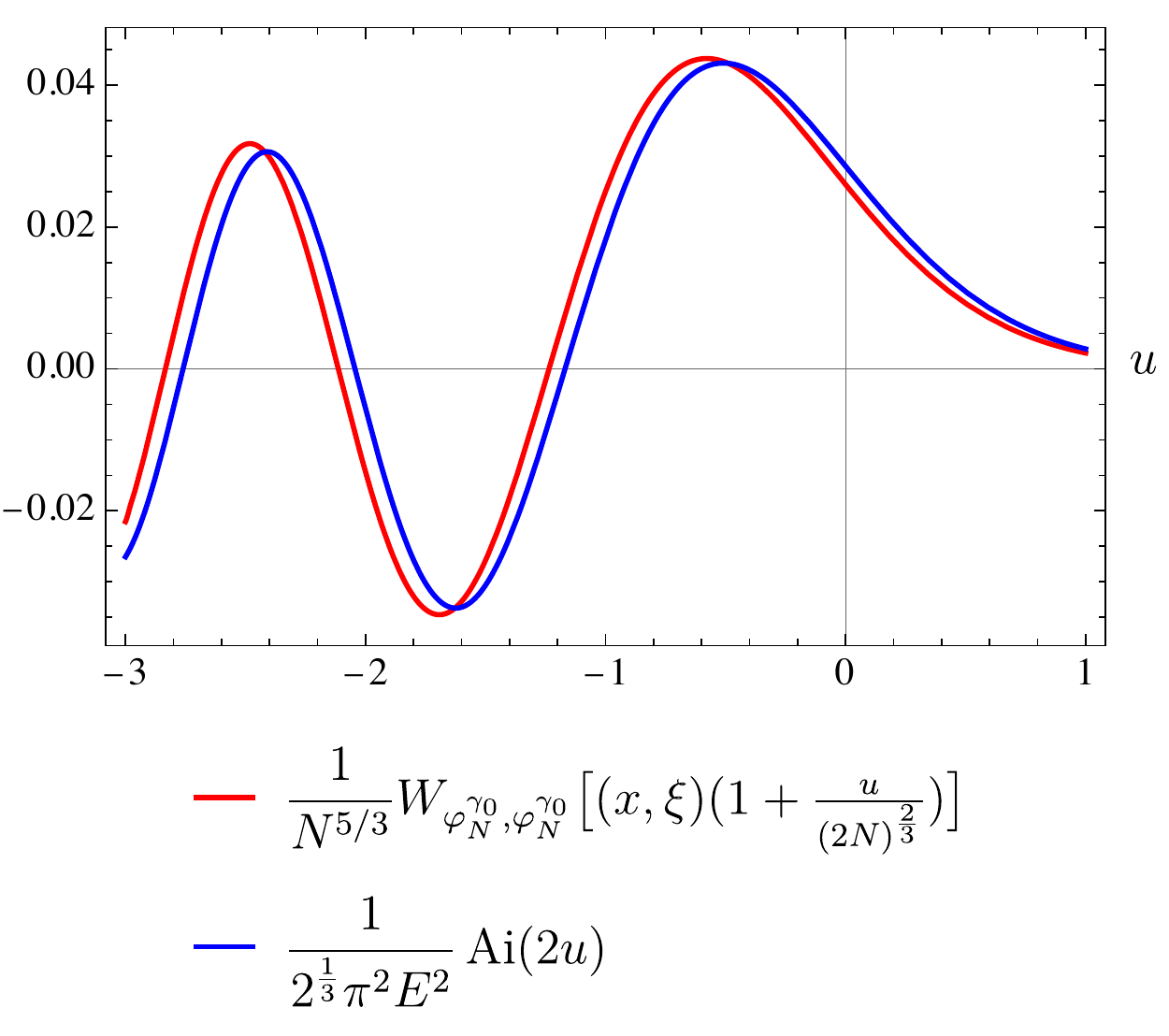} }}
    \caption{}\label{fig:1}%
    \end{figure}
    \begin{figure}[!htbp]
    \centering
    \subfloat[\centering The Wigner distribution and the scaling when $N=50$ and $E=1$]{{\includegraphics[width=8cm]{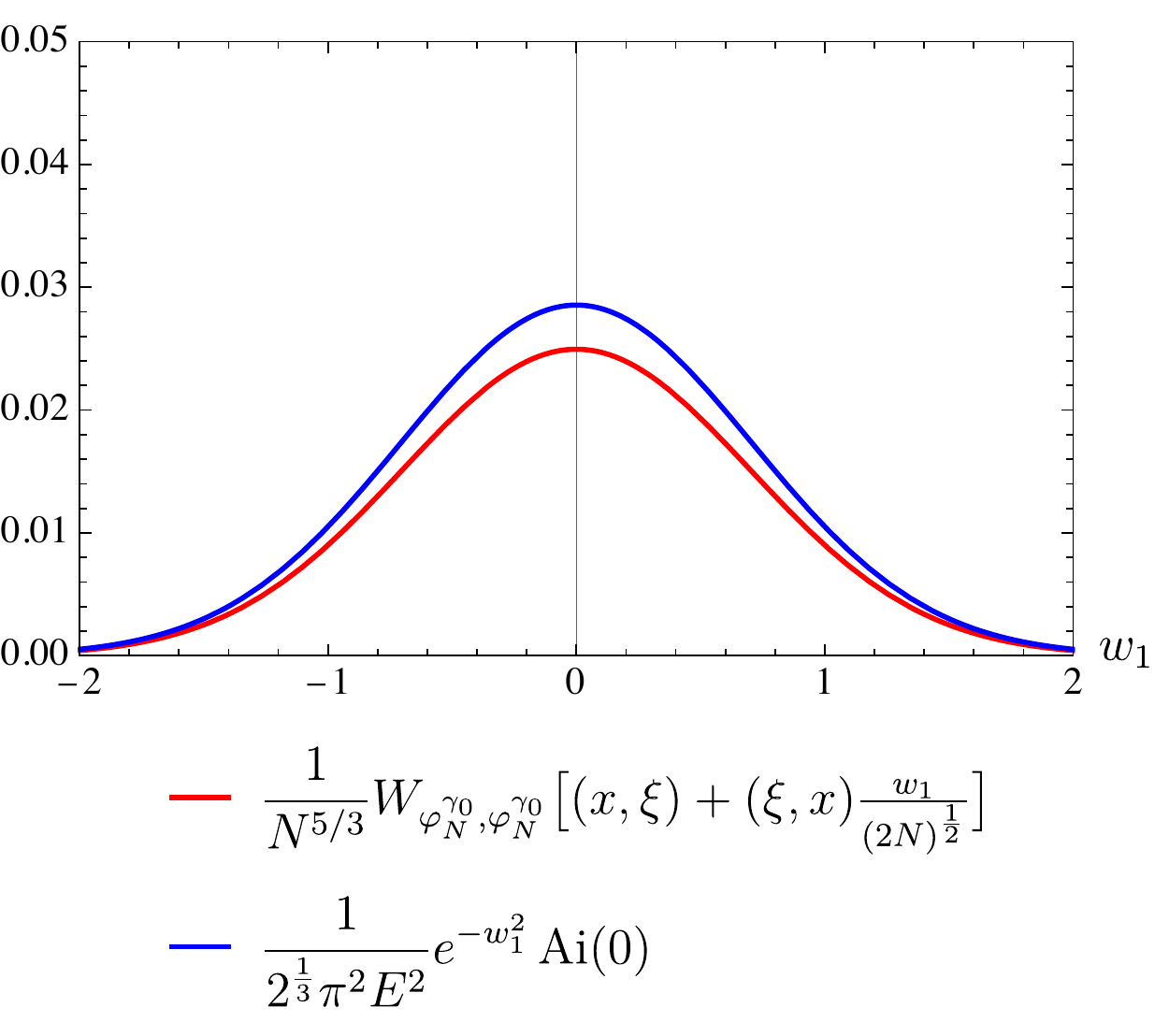} }}%
    \subfloat[\centering The Wigner distribution and the scaling when $N=200$ and $E=1$]{{\includegraphics[width=8cm]{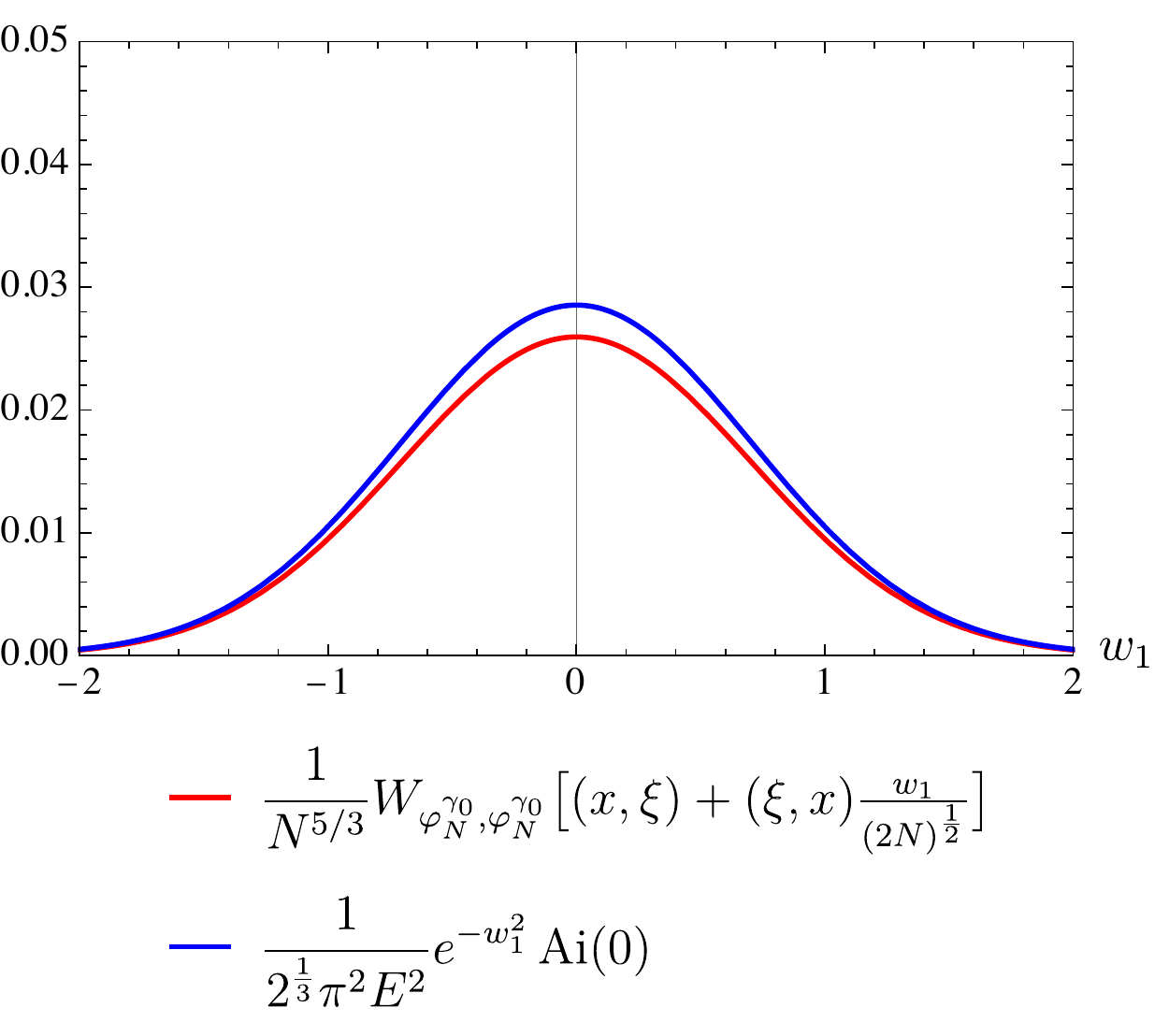} }}
    \caption{}\label{fig:2}%
    \end{figure}
\begin{rem} The
point $(x,\xi)$ lies on an orbit $\gamma_0$ which lies on the energy surface $\Sigma_E:=\{(x,\xi) :
\tfrac{1}{2}(|x|^2+|\xi|^2)=E\}$. Normal to $\Sigma_E$ at the
point $(x,\xi)$ is the vector $(x,\xi) u$ where $u \in
\mathbb{R}$, and tangent to $\Sigma_E$ at the point $(x,\xi)$
are the the vectors $w_1v_1$ and $w_2v_2$ where
$w_1,w_2 \in \mathbb{R}$. The above theorem shows that
$W_{\varphi_N^{\gamma_0},\varphi_N^{\gamma_0}}$ exhibits Airy
scaling normal to the energy surface and Gaussian scaling
tangent to the energy surface. This is consistent with the Airy
scaling result of the Wigner distributions of the eigenspace
projectors in \cite{HZ20}, and we explain this in Section \ref{s:p}.
\par The function $a(u)$ has a geometric interpretation. When
$u<0$, the function $-\frac{2E}{3}ia(u)^{3/2}$ is the area $A_u$ of the
circular segment bounded by $\gamma_0$ where $|u|\sqrt{2E}$ is the height
of the arced portion, as depicted in Figure
\ref{f:segment}. When $u \geq 0$, we have a similar
interpretation but for a hyperbolic segment, also depicted in
Figure \ref{f:segment}.

\end{rem}

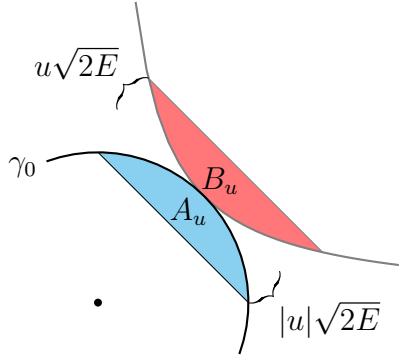
\begin{figure}[!htbp]
\centering
\begin{tikzpicture}
  \coordinate (A) at ({0.669485-2},{2.98737-2}) {};
  \coordinate (B) at ({2.98737-2},{0.669485-2}) {};
  \fill [BabyRed, domain=0.669485:2.98737, variable=\x]
  (A)
  -- plot ({\x-2}, {2/\x-2})
  -- (B)
  -- cycle;
  \fill [BabyBlue, domain=0:90, variable=\x]
  (0,-2)
  -- plot ({2*cos(\x)-2}, {2*sin(\x)-2})
  -- (-2,0)
  -- cycle;
  \draw (-2,0) -- (0,-2);
  \draw [gray] (A) -- (B); 
  \draw [gray,thick,domain=0.5:4] plot ({\x-2}, {2/\x-2});
  \draw [thick,domain=-20:110] plot ({2*cos(\x)-2}, {2*sin(\x)-2});
  \fill[fill=black] (-2,-2) circle (1.5pt);
  \draw [thick,
    decorate, 
    decoration = {calligraphic brace,
        raise=0pt,
        amplitude=5pt,
        aspect=0.5}]  ({sqrt(2)-1},{sqrt(2)-3}) -- (0,-2)
node[pos=0.5,below right=1pt,black]{$|u|\sqrt{2E}$};
\node (c) at (-3,-0.2) {$\gamma_0$};
\node (d) at ({(sqrt(2)-3)/2},{(sqrt(2)-3)/2}) {$A_u$};
\node (e) at ({(3*sqrt(2)-5)/2},{(3*sqrt(2)-5)/2}) {$B_u$};
  \draw [thick,
    decorate, 
    decoration = {calligraphic brace,
        raise=0pt,
        amplitude=5pt,
        aspect=0.5}]  ({-1.74473},{0.573136}) -- (A)
node[pos=0.5,above left=1pt,black]{$u\sqrt{2E}$};
\end{tikzpicture}
\caption{This figure shows the geometric interpretation of $a(u)$ in terms of the areas $A_u$ and $B_u$.}\label{f:segment}
  \end{figure}
    \subsection{Future Work}
    The asymptotics of the Wigner distributions in this paper are some of the few explicit results on Wigner distributions of eigenfunctions. In future work, we plan to study the asymptotics of Wigner distributions of orbital coherent states of the hydrogen atom. They are much more difficult to compute, but the isotropic harmonic oscillator of this paper provides a more readily computable model which guides future work.
    \subsection{Acknowledgments}
    This article is part of the Ph.D. thesis of the author at Northwestern University under the guidance of Steve Zelditch. The author thanks Stephan De Bi\`evre for beneficial comments. The author also thanks Michael Geis, Erik Hupp, and Abraham Rabinowitz for helpful conversations.
\section{Proofs of the Theorems}
\subsection{Preliminaries}
\subsubsection{Formula}
Before we get into the proofs of the theorems, we first write
down a formula for the Wigner distribution. Recall from \eqref{eq:coherent} that the
coherent state $\varphi_N^{\gamma_0}$ is
$$\varphi_{N}^{\gamma_0}(x)=\frac{(N+1)^{\frac{N+1}{2}}}{E^{\frac{N+1}{2}}\sqrt{\pi
    N!}}(x_1+ix_2)^Ne^{-\tfrac{(N+1)}{2 E}|x|^2}.$$
The Wigner distribution for $E=(N+1)\hbar$ is then
\begin{align}
  W_{\varphi_N^{\gamma_0},\varphi_N^{\gamma_0}}(x,\xi) &= \int_{\mathbb{R}^2}\varphi_N^{\gamma_0}(x+\tfrac{v}{2})\overline{\varphi_N^{\gamma_0}(x-\tfrac{v}{2})}e^{-\frac{i}{\hbar
   }\langle v,\xi\rangle}\frac{dv}{4\pi^2 \hbar^2} \label{eq:wigl2}\\
 &= \frac{(N+1)^{N+3}}{4\pi^3 E^{N+3} N!} \int_{\mathbb{R}^2} (x + \tfrac{v}{2})^N \overline{(x  - \tfrac{v}{2})}^N  e^{- \frac{N+1}{2E} |x - \frac{v}{2}|^2}  e^{- \frac{N+1}{2E} |x + \frac{v}{2}|^2} e^{- i\frac{(N+1)}{E} \langle v,  \xi \rangle} d v \nonumber
\end{align}
where the multiplication within the integral is complex multiplication. Since $$|x - \tfrac{v}{2}|^2 + |x + \tfrac{v}{2}|^2 =2|x|^2+\tfrac{1}{2}|v|^2, $$
we have
\begin{align}
W_{\varphi_N^{\gamma_0},\varphi_N^{\gamma_0}}(x,\xi) &= \frac{(N+1)^{N+3}}{4\pi^3E^{N+3} N!} \int_{\mathbb{R}^2} (x + \tfrac{v}{2})^N \overline{(x  - \tfrac{v}{2})}^N  e^{- \frac{N+1}{2E} |x - \frac{v}{2}|^2}  e^{- \frac{N+1}{2E} |x + \frac{v}{2}|^2} e^{- i\frac{(N+1)}{E} \langle v,  \xi \rangle} d v \nonumber\\
  &=  \frac{(N+1)^{N+3}}{4\pi^3 E^{N+3} N!}  e^{- \frac{(N+1)}{E}  |x|^2} \int_{\mathbb{R}^2} (x + \tfrac{v}{2})^N \overline{(x  - \tfrac{v}{2})}^N  e^{- \frac{N+1}{4E}   |v|^2}  
    e^{- i\frac{(N+1)}{E}\langle v, \xi \rangle} d
    v\label{eq:forlem} \\
  &=  \frac{(N+1)^{N+3}}{4\pi^3 E^{N+2} N!}  e^{- \frac{(N+1)}{E}  |x|^2} \int_{\mathbb{R}^2} (x + \tfrac{v}{2}\sqrt{E})^N \overline{(x  - \tfrac{v}{2}\sqrt{E})}^N  e^{- \frac{N+1}{4}   |v|^2}  
    e^{- i\frac{(N+1)}{\sqrt{E}}\langle v, \xi \rangle} d v,\nonumber  
\end{align}
where the last line we applied the substitution $v \to \sqrt{E}v$.
To apply the stationary phase method on this formula later on, we write
$$e^{N \log (x + \frac{v}{2}\sqrt{E}) \overline{(x  - \frac{v}{2}\sqrt{E})}}=(x + \tfrac{v}{2}\sqrt{E})^N
\overline{(x  - \tfrac{v}{2}\sqrt{E})}^N $$
where we use the principal branch of $\log$ and define the complex phase
\begin{align}
\Psi(x,\xi,v) \coloneqq& -i \log (x + \tfrac{v}{2}\sqrt{E})
                         \overline{(x  -
                         \tfrac{v}{2}\sqrt{E})}+i\log(E)+i(\tfrac{|x|^2}{E}-1)+i\tfrac{|v|^2}{4}-\tfrac{1}{\sqrt{E}}\langle
                         v,\xi \rangle\nonumber \\
  =&  -i \log (\tfrac{x}{\sqrt{E}} + \tfrac{v}{2})
                         \overline{(\tfrac{x}{\sqrt{E}}  -
                         \tfrac{v}{2})}+i(\tfrac{|x|^2}{E}-1)+i\tfrac{|v|^2}{4}-\tfrac{1}{\sqrt{E}}\langle
                         v,\xi \rangle \label{eq:phase}
\end{align}
The Wigner distribution is then
\begin{equation}\label{eq:wigstat}
W_{\varphi_N^{\gamma_0},\varphi_N^{\gamma_0}}(x,\xi)=  \frac{(N+1)^{N+3}}{4\pi^3E^{2}  N!}  e^{-N- \frac{1}{E}  |x|^2} \int_{\mathbb{R}^2}  e^{- \frac{1}{4E}   |v|^2}  
e^{- \frac{i}{E}\langle v, \xi \rangle}e^{iN \Psi(x,\xi,v)} d v.
\end{equation}
We can use Stirling's approximation to write
\begin{equation*}\label{eq:wigstat2}
W_{\varphi_N^{\gamma_0},\varphi_N^{\gamma_0}}(x,\xi)= \frac{(N+1)^{5/2}}{4\pi^3\sqrt{2\pi}E^{2}} e^{1- \frac{1}{E}  |x|^2} \big(1+O(\tfrac{1}{N})\big) \int_{\mathbb{R}^2}  e^{- \frac{1}{4}   |v|^2}  
e^{- \frac{i}{\sqrt{E}}\langle v, \xi \rangle}e^{iN \Psi(x,\xi,v)} d v,
\end{equation*}
which implies
\begin{equation}\label{eq:wigstat3}
\tfrac{1}{N^{5/3}}W_{\varphi_N^{\gamma_0},\varphi_N^{\gamma_0}}(x,\xi)= \frac{N^{5/6}}{4\pi^3\sqrt{2\pi}E^{2}}  e^{1- \frac{1}{E}  |x|^2} \big(1+O(\tfrac{1}{N})\big) \int_{\mathbb{R}^2}  e^{- \frac{1}{4}   |v|^2}  
e^{- \frac{i}{\sqrt{E}}\langle v, \xi \rangle}e^{iN \Psi(x,\xi,v)} d v.
\end{equation}
The next lemma will be very useful for our stationary phase
approximations.
\begin{lem}\label{l:pos}
For $\Psi$ defined in \eqref{eq:phase}, we have
$$\Im \Psi(x,\xi,v) \geq 0 $$
with equality if and only if $|v|^2=4(1-\tfrac{1}{E}|x|^2)$ and $\langle x,v\rangle=0$.
\end{lem}
\begin{proof}
Indeed,
\begin{align*}
\Im \Psi(x,\xi,v) &\overset{\mathclap{\eqref{eq:phase}}}{=}- \log\Big| (\tfrac{x}{\sqrt{E}} + \tfrac{v}{2})                        \overline{(\tfrac{x}{\sqrt{E}}  -                    \tfrac{v}{2})}\Big|+\tfrac{|x|^2}{E}-1+\tfrac{|v|^2}{4}\\
  &=- \log\big| \tfrac{|x|^2}{E}-\tfrac{|v|^2}{4} +
    i\tfrac{1}{\sqrt{E}}\Im(v \overline{x})
    \big|+\tfrac{|x|^2}{E}-1+\tfrac{|v|^2}{4}\\
  &=-\tfrac{1}{2} \log\Big( \big(\tfrac{|x|^2}{E}-\tfrac{|v|^2}{4}\big)^2 +
    \tfrac{1}{E}\big(\Im(v
    \overline{x})\big)^2\Big)+\tfrac{|x|^2}{E}-1+\tfrac{|v|^2}{4}\\
  &\geq - \tfrac{1}{2}\log\Big(
    \big(\tfrac{|x|^2}{E}-\tfrac{|v|^2}{4}\big)^2 +
    \tfrac{1}{E}|xv|^2\Big)+\tfrac{|x|^2}{E}-1+\tfrac{|v|^2}{4}\\
  &= - \log\big(\tfrac{|x|^2}{E}+\tfrac{|v|^2}{4}
    \big)+\tfrac{|x|^2}{E}-1+\tfrac{|v|^2}{4}\\
  &=- \log\big(1+\tfrac{|x|^2}{E}-1+\tfrac{|v|^2}{4}
    \big)+\tfrac{|x|^2}{E}-1+\tfrac{|v|^2}{4}\\
  &\geq 0,
\end{align*}
where the first inequality we use the fact that $|\Im z | \leq |z|$ (with equality if and only if $\Re z =0$), and the last inequality we use the fact that $-\log(1+x)+x \geq 0$ for all $x \geq -1$ (with equality if and only if $x=0$), which completes the proof.
\end{proof}
\subsubsection{Symmetries}
One important property of Wigner distributions is the so-called
metaplectic covariance. In dimension 2, it states that given $f
\in L^2(\mathbb{R}^2)$ and $\mu$ the metaplectic representation
on $L^2(\mathbb{R}^2)$, we have
\begin{equation}\label{eq:meta}
W_{\mu(A)f,\mu(A)f}(x,\xi) = W_{f,f}\big(A^*(x,\xi)\big)
\end{equation}
for all $A \in \operatorname{Sp}(4,\mathbb{R})$ (see page 180 of \cite{F89}). This is useful
in proving results for the Wigner distributions of general
coherent states $\varphi_N^{\gamma}$ from that of Wigner distributions of
$\varphi_N^{\gamma_0}$. In particular, if $U \in
\operatorname{SU}(2) \subset \operatorname{Sp}(4,\mathbb{R})$ is
such that $U \cdot \gamma_0=\gamma$, then we have
\begin{equation}\label{eq:meta2}
W_{\varphi_N^{\gamma},\varphi_N^{\gamma}}(x,\xi) \overset{\mathclap{\eqref{eq:meta}}}{=} W_{\mu(U)\varphi_N^{\gamma_0},\mu(U)\varphi_N^{\gamma_0}}(x,\xi) \overset{\mathclap{\eqref{eq:genphi}}}{=} W_{\varphi_N^{\gamma_0},\varphi_N^{\gamma_0}}\big(U^*(x,\xi)\big).
\end{equation}
The next lemma will be useful in proving our results for general points in phase space from an initial point.
\begin{lem}\label{lem:rotation} For $g_{\theta}=\begin{psmallmatrix}\cos \theta & -\sin \theta \\
\sin \theta & \cos \theta \end{psmallmatrix} \in \operatorname{SO}(2)$, we have
$$W_{\varphi_N^{\gamma_0},\varphi_N^{\gamma_0}}(g_{\theta}\cdot
x,g_{\theta}\cdot
\xi)=W_{\varphi_N^{\gamma_0},\varphi_N^{\gamma_0}}(x,\xi).$$
In particular, $W_{\varphi_N^{\gamma},\varphi_N^{\gamma}}$ is radial in the plane containing $\gamma$.
\end{lem}
\begin{proof}
One proof is to use the fact that the restricted representation of the metaplectic representation to the subgroup $\operatorname{SO}(2)$ is rotation of functions. The coherent state $\varphi_N^{\gamma_0}$ is invariant under rotation up to a phase factor, which results in equal Wigner distributions. Alternatively,
we can use formula \eqref{eq:forlem}. We have
\begin{align*}
W_{\varphi_N^{\gamma_0},\varphi_N^{\gamma_0}}&(g_{\theta}\cdot
x,g_{\theta}\cdot
\xi) =  \frac{(N+1)^{N+3}}{4\pi^3 E^{N+3} N!}  e^{-
     \frac{(N+1)}{E}
                                               |g_{\theta}\cdot
                                               x|^2} \\
  &\qquad \qquad \qquad \qquad \int_{\mathbb{R}^2} (g_{\theta}\cdot x + \tfrac{v}{2})^N \overline{(g_{\theta}\cdot x  - \tfrac{v}{2})}^N  e^{- \frac{N+1}{4E}   |v|^2}  
    e^{- i\frac{(N+1)}{E}\langle v,g_{\theta}\cdot \xi \rangle}
  d v\\
  &=  \frac{(N+1)^{N+3}}{4\pi^3 E^{N+3} N!}  e^{- \frac{(N+1)}{E}  | x|^2} \int_{\mathbb{R}^2} (g_{\theta}\cdot x + \tfrac{v}{2})^N \overline{(g_{\theta}\cdot x  - \tfrac{v}{2})}^N  e^{- \frac{N+1}{4E}   |v|^2}  
    e^{- i\frac{(N+1)}{E}\langle g_{\theta}^*\cdot v,\xi
    \rangle} d v\\
   &=  \frac{(N+1)^{N+3}}{4\pi^3 E^{N+3} N!}  e^{- \frac{(N+1)}{E}  | x|^2} \int_{\mathbb{R}^2} (g_{\theta}\cdot x + \tfrac{g_{\theta}\cdot v}{2})^N \overline{(g_{\theta}\cdot x  - \tfrac{g_{\theta}\cdot v}{2})}^N  e^{- \frac{N+1}{4E}   |g_{\theta}\cdot v|^2}  
     e^{- i\frac{(N+1)}{E}\langle  v,\xi \rangle} d v\\
  &=  \frac{(N+1)^{N+3}}{4\pi^3 E^{N+3} N!}  e^{- \frac{(N+1)}{E}  | x|^2} \int_{\mathbb{R}^2}\big( g_{\theta}\cdot( x + \tfrac{v}{2})\big)^N \overline{\big(g_{\theta}\cdot (x  - \tfrac{ v}{2})\big)}^N  e^{- \frac{N+1}{4E}   |v|^2}  
    e^{- i\frac{(N+1)}{E}\langle  v,\xi \rangle} d v\\
   &=  \frac{(N+1)^{N+3}}{4\pi^3 E^{N+3} N!}  e^{- \frac{(N+1)}{E}  | x|^2} \int_{\mathbb{R}^2}e^{iN \theta}( x + \tfrac{v}{2})^N e^{-iN \theta}\overline{(x  - \tfrac{ v}{2})}^N  e^{- \frac{N+1}{4E}   |v|^2}  
     e^{- i\frac{(N+1)}{E}\langle  v,\xi \rangle} d v\\
  &=W_{\varphi_N^{\gamma_0},\varphi_N^{\gamma_0}}(x,\xi),
\end{align*}
where the third line we switch variables $v \to g_{\theta}\cdot v$. The last statement of the lemma follows from an application of \eqref{eq:meta2}.
\end{proof}

\subsection{Proof of Proposition \ref{t:delta}}
We first prove this for $\gamma_0(t)=\sqrt{E}( \cos t ,  \sin t, - \sin t, \cos t)$.
We write
$$I_N \coloneqq \int_{T^*\mathbb{R}^2}a(x,\xi)W_{\varphi_N^{\gamma_0},\varphi_N^{\gamma_0}}(x,\xi)dxd\xi. $$
Using \eqref{eq:wigstat3}, we have
\begin{align*}
I_N &= \frac{N^{5/2}}{4\pi^3\sqrt{2\pi}E^{2}}\big(1+O(\tfrac{1}{N})\big)  \int_{\mathbb{R}^6} a(x,\xi)e^{1- \frac{1}{E}  |x|^2}  e^{- \frac{1}{4}   |v|^2}  
e^{- \frac{i}{\sqrt{E}}\langle v, \xi \rangle}e^{iN \Psi(x,\xi,v)} d vdxd\xi.
\end{align*}
We now apply the method of stationary phase for Bott-Morse
functions (see Theorem \ref{t:BM} in Section \ref{s:phase}). Recall from
\eqref{eq:phase} that $\Psi$ is defined by
$$\Psi(x,\xi,v) \coloneqq -i \log (\tfrac{x}{\sqrt{E}} + \tfrac{v}{2})
                         \overline{(\tfrac{x}{\sqrt{E}}  -
                         \tfrac{v}{2})}+i(\tfrac{|x|^2}{E}-1)+i\tfrac{|v|^2}{4}-\tfrac{1}{\sqrt{E}}\langle
                         v,\xi \rangle.$$ 
Using Wirtinger derivatives in $\xi$ and $v$ along with Lemma \ref{l:pos},
one can deduce
\begin{align*}
C \coloneqq \{(x,\xi,v) \in \mathbb{R}^6=\mathbb{C}^3\mid \Im\Psi=0,\Psi'=0\}=\{(\sqrt{E}e^{i\theta},i\sqrt{E}e^{i\theta},0) \in \mathbb{C}^3\mid \theta \in [0,2\pi)\}.
\end{align*}
This is a one-dimensional submanifold where the tangent space
$T_{\theta}C$ is spanned by the vector
$$-\sin
(\theta)\partial_{x_1}+\cos(\theta)\partial_{x_2}-\cos(\theta)\partial_{\xi_1}-\sin(\theta)\partial_{\xi_2}. $$
Completing this to an orthonormal basis of $\mathbb{R}^6$, we
see that normal
space, $N_{\theta}C$, of $C$ at $\theta$ is the span of
$$\partial_{v_1},\partial_{v_2},\tfrac{1}{\sqrt{2}}(\partial_{x_2}+\partial_{\xi_1}),\tfrac{1}{\sqrt{2}}(\partial_{x_1}-\partial_{\xi_2}),\tfrac{1}{\sqrt{2}}\big(\cos
(\theta)\partial_{x_1}+\sin(\theta)\partial_{x_2}-\sin(\theta)\partial_{\xi_1}+\cos(\theta)\partial_{\xi_2}\big).$$
Using this basis, one can compute
$$\det(-i\Psi''|_{N_{\theta}C})=8/E^3.$$
It is easy to see that the natural Riemannian measure on $C$ is
$d\mu_C=\sqrt{2E}d\theta$. Now applying Theorem \ref{t:BM}, we
see
\begin{align*}
I_N&= \frac{N^{5/2}}{4\pi^3\sqrt{2\pi}E^{2}}\big(1+O(\tfrac{1}{N})\big)  \int_{\mathbb{R}^6} a(x,\xi)e^{1- \frac{1}{E}  |x|^2}  e^{- \frac{1}{4}   |v|^2}  
e^{- \frac{i}{\sqrt{E}}\langle v, \xi \rangle}e^{iN
     \Psi(x,\xi,v)} d vdxd\xi\\
  &=
    \frac{N^{5/2}}{4\pi^3\sqrt{2\pi}E^{2}}\big(1+O(\tfrac{1}{N})\big)\Big[
    \Big(\frac{2\pi}{N}\Big)^{\frac{5}{2}} \int_{C}
    \frac{a(u)}{\sqrt{8/E^3}}
    d\mu_C(u)+O(N^{-\frac{7}{2}})\Big]\\
   &=\frac{1}{2\pi}\int_{0}^{2\pi}a(\sqrt{E}e^{i\theta},i\sqrt{E}e^{i\theta})d\theta+O(\tfrac{1}{N})\\
  &=\frac{1}{2\pi}\int_{0}^{2\pi}a\big(\gamma_0(\theta)\big)d\theta+O(\tfrac{1}{N}),
\end{align*}
which is the desired result for $\gamma_0$. Another way to prove
this is to initially change $x$ to polar coordinates
$re^{i\theta}$, and then use Theorem \ref{t:hor} in the variables
$r,\xi,v$ for fixed $\theta$.\\
To prove the
proposition for a general orbit $\gamma$, let $U \in
\operatorname{SU}(2)$ be such that $U \cdot \gamma_0=\gamma$,
and let $a:T^*\mathbb{R}^2 \to \mathbb{R}$ be a smooth function
that decays exponentially. Then $b(x,\xi) \coloneqq a\big(
U(x,\xi)\big)$ is also a smooth function with exponential decay. Applying what we
proved to $b$, we have
\begin{equation}\label{eq:mid}
\int_{T^*\mathbb{R}^2}b(x,\xi)W_{\varphi_N^{\gamma_0},\varphi_N^{\gamma_0}}(x,\xi)\
dxd\xi \xrightarrow{N \to \infty}
\frac{1}{2\pi}\int_0^{2\pi}b\big(\gamma_0(\theta)\big)\ d\theta.
\end{equation}
Changing variables and using \eqref{eq:meta2}, we have
\begin{align*}
\int_{T^*\mathbb{R}^2}b(x,\xi)W_{\varphi_N^{\gamma_0},\varphi_N^{\gamma_0}}(x,\xi)\
dxd\xi
  &=\int_{T^*\mathbb{R}^2}a\big(U(x,\xi)\big)W_{\varphi_N^{\gamma_0},\varphi_N^{\gamma_0}}(x,\xi)dxd\xi\\
  &=\int_{T^*\mathbb{R}^2}a(x,\xi)W_{\varphi_N^{\gamma_0},\varphi_N^{\gamma_0}}\big(U^*(x,\xi)\big)dxd\xi\\
  &=\int_{T^*\mathbb{R}^2}a(x,\xi)W_{\mu(U)\varphi_N^{\gamma_0},\mu(U)\varphi_N^{\gamma_0}}(x,\xi)dxd\xi\\
   &=\int_{T^*\mathbb{R}^2}a(x,\xi)W_{\varphi_N^{\gamma},\varphi_N^{\gamma}}(x,\xi)dxd\xi.
\end{align*}
Substituting this into \eqref{eq:mid}, we see
$$\int_{T^*\mathbb{R}^2}a(x,\xi)W_{\varphi_N^{\gamma},\varphi_N^{\gamma}}(x,\xi)dxd\xi \xrightarrow{N \to \infty}
\frac{1}{2\pi}\int_0^{2\pi}a\big(\gamma(\theta)\big)\ d\theta ,$$
as desired.
\subsection{Proof of Theorem \ref{t:main2}}
We first prove this theorem for $\gamma_0(t)=\sqrt{E}( \cos t ,  \sin t, - \sin t, \cos t)$ by applying Theorem \ref{t:horb} to \eqref{eq:wigstat}. Before we do that, we first compute the Wigner distribution explicitly when $(x,\xi)=(0,0)$
\subsubsection{Proof of (4)}
Using \eqref{eq:wigl2} and the fact that $\lVert \varphi_N^{\gamma_0}\rVert_2=1$,
we have
\begin{align}
 W_{\varphi_N^{\gamma_0},\varphi_N^{\gamma_0}}(0,0) &=
                                                      \int_{\mathbb{R}^2}\varphi_N^{\gamma_0}(\tfrac{v}{2})\overline{\varphi_N^{\gamma_0}(-\tfrac{v}{2})}\frac{dv}{4\pi^2
                                                      \hbar^2}\nonumber \\
                                                    &=\frac{(-1)^N
                                                      (N+1)}{\pi^2E^2}\int_{\mathbb{R}^2}|\varphi_N^{\gamma_0}(v)|^2dv\\
  &=\frac{(-1)^N (N+1)}{\pi^2E^2}\label{eq:0}
\end{align}
where $E=(N+1)\hbar$. This shows (4) in the statement of the theorem.
\subsubsection{Proof of (3)} Now we apply the
method of stationary phase to \eqref{eq:wigstat3}. Recall
the complex phase defined in \eqref{eq:phase}:
$$\Psi(x,\xi,v) \coloneqq -i \log (\tfrac{x}{\sqrt{E}} + \tfrac{v}{2})
                         \overline{(\tfrac{x}{\sqrt{E}}  -
                         \tfrac{v}{2})}+i(\tfrac{|x|^2}{E}-1)+i\tfrac{|v|^2}{4}-\tfrac{1}{\sqrt{E}}\langle
                         v,\xi \rangle.$$ 
Computing the Wirtinger derivatives in $v$, we have the following lemma:
\begin{lem}\label{l:1} The critical point equation $D_v\Psi(x,\xi,v)=0$ is 
\begin{equation}\label{eq:1}
\begin{cases}
\partial_v \Psi  = 0\\
\overline{\partial}_v \Psi =0
\end{cases}
\iff\begin{cases}
i E +
\tfrac{i}{2} v\sqrt{E}\overline{(x  - \tfrac{v}{2}\sqrt{E})}  -\xi\overline{(x  - \tfrac{v}{2}\sqrt{E})}  = 0\\
- i E  + \tfrac{i}{2} \bar{v}\sqrt{E} (x + \tfrac{v}{2}\sqrt{E})  -
\overline{\xi}(x + \tfrac{v}{2}\sqrt{E}) =0
\end{cases},
\end{equation}
which has solutions only if $\langle x,\xi \rangle =0$.
\begin{proof}
Using $\langle v, \xi \rangle=v_1 \xi_1 + v_2 \xi_2 = \Re  v \overline{\xi}=  \tfrac{1}{2}(v \overline{\xi} +\overline{v}\xi)$, 
we have $\overline{\partial}_v (v_1 \xi_1 + v_2 \xi_2) =\frac{1}{2}\xi $ and
$\partial_v(v_1\xi_1+v_2\xi_2)=\frac{1}{2}\overline{\xi}$. Then
have
$$\overline{\partial}_v \Psi =  \frac{i}{2} \frac{\tfrac{x}{\sqrt{E}} + \tfrac{v}{2}}{ (\tfrac{x}{\sqrt{E}} + \tfrac{v}{2})\overline{(\tfrac{x}{\sqrt{E}} - \tfrac{v}{2})}} + \tfrac{i}{4 } v  - \tfrac{1}{2 \sqrt{E}}\xi  = 0, $$
and
$$\partial_v \Psi =- \frac{i}{2} \frac{\overline{\tfrac{x}{\sqrt{E}} - \tfrac{v}{2}} }{ (\tfrac{x}{\sqrt{E}} + \tfrac{v}{2})\overline{(\tfrac{x}{\sqrt{E}} - \tfrac{v}{2})}} + \tfrac{i}{4 } \bar{v}   -
\tfrac{1}{2\sqrt{ E}}\overline{\xi}  = 0.$$
Canceling and multiplying through, we have the desired system. Adding these two equations, we see
\begin{equation}\label{eq:2}
i\sqrt{E} [\Re(v \overline{x})- \Im(v \overline{\xi})]=2 \Re(\xi
\overline{x}).
\end{equation}
In particular, both sides of this equation are $0$. We then see
$\langle x,\xi \rangle=0.$
\end{proof}
\end{lem}
We now characterize when the set $C_{x,\xi}\coloneqq \{v\in
\mathbb{R}^2 \mid D_v \Psi(x,\xi,v)=0,\Im
\Psi(x,\xi,v)=0\}$ is nonempty.
\begin{lem}\label{lem:crityes}
The set $C_{x,\xi}\coloneqq \{v\in
\mathbb{R}^2 \mid D_v \Psi(x,\xi,v)=0,\Im
\Psi(x,\xi,v)=0\}$ is nonempty if and only if $0 \leq |x| \leq
\sqrt{E}$ and $\xi=ix$. In this case, when $x \neq 0$,
$$C_{x,\xi}=\Big\{v_{\pm}\coloneqq \pm \tfrac{2}{|x|}ix\sqrt{1-\tfrac{|x|^2}{E}}\Big\}. $$
\end{lem}
\begin{proof}
By Lemma \ref{l:1}, $D_v \Psi(x,\xi,v)=0$ has solutions only if $\langle
x,\xi \rangle=0$. This means either $x = 0$ or $\xi = k i x$ for
some $k \in \mathbb{R}$. We consider each case separately.\\
{$\pmb {x=0}$}: In the case $x=0$, the top equation of \eqref{eq:1} becomes
\begin{equation}\label{eq:new}
iE -\tfrac{i}{4}E|v|^2 +\sqrt{E}\xi\overline{v}/2  = 0
\end{equation}
By the equality statement of Lemma \ref{l:pos}, we have $|v|=2$. With \eqref{eq:new}, this implies $\xi \overline{v}=0$ and hence $\xi=0$.\\ 
{ $\pmb{\xi=kix \text{ and } x\neq 0}$}: In the case $\xi=kix$ and $x \neq 0$, \eqref{eq:1} becomes
\begin{equation}\label{eq:4}
\begin{cases}
\tfrac{\sqrt{E}}{2}(v \overline{x}+k x \overline{v})  = \tfrac{E}{4}|v|^2+k|x|^2-E\\
\tfrac{\sqrt{E}}{2}(\overline{v}x+k \overline{x}v)=E-k|x|^2-\tfrac{E}{4}|v|^2=-( \tfrac{E}{4}|v|^2+k|x|^2-E).
\end{cases}
\end{equation}
Taking the complex conjugate of the second equation and
comparing with the first, we see 
\begin{equation}\label{eq:inter}
\tfrac{1}{4}|v|^2+\tfrac{k}{E}|x|^2-1=0
\end{equation}
The equality statement of Lemma \ref{l:pos} forces $k=1$, and we are done. The last statement of the lemma again follows from the equality statement of Lemma \ref{l:pos}.
\end{proof}
Note that Lemma \ref{l:pos} and Lemma \ref{lem:crityes} along
with Theorem 7.7.1 of \cite{H03} implies that
$W_{\varphi_N^{\gamma_0},\varphi_N^{\gamma_0}}(x,\xi)=O(N^{-\infty})$
when $|x| \notin [0,\sqrt{E}]$ or $\xi\neq ix$, which proves (3)
in the statement of the theorem.
\subsubsection{Proof of (2)}
We next compute the Hessian of $\Psi$ in $v$ at the critical points $v_{\pm}$.
\begin{lem}\label{lem:hess} Suppose $0<|x| \leq \sqrt{E}$ and $\xi=ix$. Then the determinant of the Hessian $D_v^2\Psi$
evaluated at the critical points $v_{\pm}=\pm \frac{2}{|x|}ix\sqrt{1-\frac{|x|^2}{E}}$ is 
$$\det D_v^2\Psi(x,ix,v_{\pm})=2\tfrac{|x|^4}{E^2}-2\tfrac{|x|^2}{E} \mp i(2\tfrac{|x|^2}{E}-1) \tfrac{|x|}{\sqrt{E}}\sqrt{1-\tfrac{|x|^2}{E}}$$
with
$$|\det D_v^2\Psi(x,ix,v_{\pm})|=\tfrac{|x|}{\sqrt{E}}\sqrt{1-\tfrac{|x|^2}{E}}.$$
\end{lem}
\begin{proof}
Recall from \eqref{eq:phase} that the phase function is
$$\Psi(x,\xi,v) =  -i \log (\tfrac{x}{\sqrt{E}} + \tfrac{v}{2})
                         \overline{(\tfrac{x}{\sqrt{E}}  -
                         \tfrac{v}{2})}+i(\tfrac{|x|^2}{E}-1)+i\tfrac{|v|^2}{4}-\tfrac{1}{\sqrt{E}}\langle
                         v,\xi \rangle $$
with first derivatives
\begin{align*}
  \partial_v \Psi =- \frac{i}{2} \frac{1 }{ (\tfrac{x}{\sqrt{E}} + \tfrac{v}{2})} + \tfrac{i}{4 } \bar{v}   -
\tfrac{1}{2\sqrt{ E}}\overline{\xi},\qquad 
\overline{\partial}_v \Psi =  \frac{i}{2} \frac{1}{\overline{(\tfrac{x}{\sqrt{E}} - \tfrac{v}{2})}} + \tfrac{i}{4 } v  - \tfrac{1}{2 \sqrt{E}}\xi  = 0.
\end{align*}
The second derivatives are then
\begin{align*}
\partial_v^2\Psi =\frac{i}{4} \frac{1}{ (\tfrac{x}{\sqrt{E}} + \tfrac{v}{2})^2},\quad 
  \overline{\partial}_v\partial_v \Psi = \partial_v \overline{\partial}_v
                     \Psi= \tfrac{i}{4},\quad
  \overline{\partial}_v^2\Psi = \frac{i}{4} \frac{1}{ \overline{(\tfrac{x}{\sqrt{E}} - \tfrac{v}{2})}^2}.
\end{align*}
The complex Hessian is
\begin{equation}\label{eq:hess}
H_{\Psi} = \frac{i}{4}\begin{pmatrix}  (\tfrac{x}{\sqrt{E}} + \tfrac{v}{2})^{-2} & 1 \\
1 & \overline{(\tfrac{x}{\sqrt{E}} -\tfrac{v}{2})}^{-2} \end{pmatrix} \implies \det H_{\Psi}\Psi = \frac{1}{16}\bigg(1-\frac{1}{\big[(\frac{x}{\sqrt{E}} +\frac{v}{2})\overline{(\frac{x}{\sqrt{E}} -\frac{v}{2})}\big]^2}\bigg).
\end{equation}
Note that we have
\begin{align*}
 (\tfrac{x}{\sqrt{E}} +\tfrac{v_{\pm}}{2})\overline{(\tfrac{x}{\sqrt{E}} -\tfrac{v_{\pm}}{2})} 
  &= \big(\tfrac{x}{\sqrt{E}} \pm \tfrac{1}{|x|}ix\sqrt{1-\tfrac{|x|^2}{E}}\big)\overline{\big(\tfrac{x}{\sqrt{E}}  \mp  \tfrac{1}{|x|}ix\sqrt{1-\tfrac{|x|^2}{E}}\big)}\\
  &= \big(\tfrac{x}{\sqrt{E}} \pm \tfrac{1}{|x|}ix\sqrt{1-\tfrac{|x|^2}{E}}\big)\big(\tfrac{\overline{x}}{\sqrt{E}}  \pm  \tfrac{1}{|x|}i \overline{x}\sqrt{1-\tfrac{|x|^2}{E}}\big)\\
  &=2\tfrac{|x|^2}{E}-1 \pm 2i \tfrac{|x|}{\sqrt{E}}\sqrt{1-\tfrac{|x|^2}{E}},
\end{align*}
which has unit norm. Squaring, we have
\begin{align*}
\big[(\tfrac{x}{\sqrt{E}} +\tfrac{v_{\pm}}{2})\overline{(\tfrac{x}{\sqrt{E}} -\tfrac{v_{\pm}}{2})} \big]^2 &=\big[
               2\tfrac{|x|^2}{E}-1 \pm 2i \tfrac{|x|}{\sqrt{E}}\sqrt{1-\tfrac{|x|^2}{E}}\big]^2\\
  &=8\tfrac{|x|^4}{E^2}-8\tfrac{|x|^2}{E}+1 \pm i4(2\tfrac{|x|^2}{E}-1) \tfrac{|x|}{\sqrt{E}}\sqrt{1-\tfrac{|x|^2}{E}}
\end{align*}
The determinant of the complex Hessian is
$$\det H_{\Psi}(x,ix,v_{\pm})=\tfrac{1}{16}\Big(8\tfrac{|x|^2}{E}-8\tfrac{|x|^4}{E^2} \pm i4(2\tfrac{|x|^2}{E}-1) \tfrac{|x|}{\sqrt{E}}\sqrt{1-\tfrac{|x|^2}{E}}\Big)$$
with
$$|\det
H_{\Psi}(x,ix,v_{\pm})|=\tfrac{|x|}{4\sqrt{E}}\sqrt{1-\tfrac{|x|^2}{E}}.$$
Note that the determinant of the real Hessian $\det D_v^2\Psi$ differs from the determinant of the complex Hessian $\det H_{\Psi}$ by a factor of $-4$, and we are done
\end{proof}
Now we use Theorem \ref{t:horb} applied to
$W_{\varphi_N^{\gamma_0},\varphi_N^{\gamma_0}}(x,\xi)$ when
$0<|x|<\sqrt{E}$ and $\xi=ix$. Indeed, the only critical points
of the phase are $v_{\pm}$ (see Lemma
\ref{lem:crityes}). Theorem \ref{t:horb} guarantees
\begin{align*}
&\Big|\tfrac{1}{N^{5/3}}W_{\varphi_N^{\gamma_0},\varphi_N^{\gamma_0}}(x,\xi)\Big| \overset{\mathclap{\eqref{eq:wigstat3}}}{=}\bigg| \frac{N^{5/6}}{4\pi^3\sqrt{2\pi}E^{2}}  e^{1- \frac{1}{E}  |x|^2} \big(1+O(\tfrac{1}{N})\big) \int_{\mathbb{R}^2}  e^{- \frac{1}{4}   |v|^2}  
e^{- \frac{i}{\sqrt{E}}\langle v, \xi \rangle}e^{iN
                                                                   \Psi(x,\xi,v)}
                                                                                d v\bigg|\\
  &\leq 2\frac{N^{5/6}}{4\pi^3\sqrt{2\pi}E^{2}}  e^{1-
    \frac{1}{E}  |x|^2} \big(1+O(\tfrac{1}{N})\big)
    \Big(\tfrac{2\pi}{N\sqrt{|\det(D_v^2\Phi(x,ix,v_+))|}}|f(v_+)|+
    O(N^{-3/2})\Big)  \\
  &=O(N^{-1/6}),
\end{align*}
as desired for (2) in the theorem.
\subsubsection{Proof of (1)} This follows directly from Theorem \ref{t:scale} by setting $u=w_1=w_2=0$. An independent proof is possible with a careful reading of the coefficients of the proof of Proposition 3 on page 334
of \cite{S93}.\\
For general $\gamma$, we use \eqref{eq:meta2}, and we are done.
\subsection{Proof of Theorem \ref{t:scale}}\label{s:scale}
Again, we start by proving this result for \\ $\gamma_0(t)=\sqrt{E}( \cos t ,  \sin t, - \sin t, \cos t)$ and
$(x,\xi)=\sqrt{E}((1,0),(0,1))$, which in complex notation is
$\sqrt{E}(1,i)$. That is, we would like to compute the limit as $N \to \infty$ of
\begin{align*}
  &\frac{1}{N^{5/3}}W_{\varphi_N^{\gamma_0},\varphi_N^{\gamma_0}}\big[\sqrt{E}(1,i)(1+u)+w_1v_1+w_2v_2\big]\\
  &=\frac{1}{N^{5/3}}W_{\varphi_N^{\gamma_0},\varphi_N^{\gamma_0}}\big[\sqrt{E}+u\sqrt{E}+iw_1\sqrt{E}+w_2\sqrt{E},i\sqrt{E}+iu\sqrt{E}+w_1\sqrt{E}-iw_2\sqrt{E}\big].
\end{align*}
where $v_1 \coloneqq \sqrt{E}(i,1)$ and $v_2 \coloneqq \sqrt{E}(1,-i)$.
We define
\begin{align*}
  x_{u,w_1,w_2} &\coloneqq \sqrt{E}+u\sqrt{E}+iw_1\sqrt{E}+w_2\sqrt{E}\\
  \xi_{u,w_1,w_2}& \coloneqq i\sqrt{E}+iu\sqrt{E}+w_1\sqrt{E}-iw_2\sqrt{E}.
\end{align*}
By \eqref{eq:wigstat3}, we have
\begin{align}
\frac{1}{N^{5/3}}W_{\varphi_N^{\gamma_0},\varphi_N^{\gamma_0}}(x_{u,w_1,w_2}, \xi_{u,w_1,w_2})&=
                                                          \frac{N^{5/6}}{4\pi^3\sqrt{2\pi}E^{2}}
                                                  e^{1-
                                                          \frac{1}{E}
                                                          |x_{u,w_1,w_2}|^2}
                                                          \big(1+O(\tfrac{1}{N})\big)
  \nonumber\\
                                                        &\qquad                                                         
                                                          \int_{\mathbb{R}^2}
                                             f(v_1,v_2,u,w_1,w_2)e^{iN \Psi(v_1,v_2,u,w_1,w_2)}
                                                          dv_1dv_2.\label{eq:v1}
\end{align}
where
\begin{align}
\Psi(v_1,v_2,u,w_1,w_2)&\coloneqq 
      -i\log\big[(1+u+w_2)^2+w_1^2-\tfrac{v_1^2}{4}-\tfrac{v_2^2}{4}+i\big(v_2(1+u+w_2)-v_1w_1\big) \big]\nonumber \\
  & +i\big((1+u+w_2)^2+w_1^2-1+\tfrac{v_1^2}{4}+\tfrac{v_2^2}{4}\big)-\big(v_1w_1+v_2(1+u-w_2)\big)\\
           f(v_1,v_2,u,w_1,w_2) &\coloneqq\ e^{-\frac{1}{4}(v_1^2+v_2^2)-i[v_1w_1+v_2(1+u-w_2)]}.   \label{eq:amp}                         
\end{align}
By Lemma \ref{l:pos} and a standard calculation, we have
$$
  \Im\Psi(0)=0,\quad 
  \Im\Psi \geq 0,\quad
  \partial_{v_1}\Psi(0)=0,\quad
  \partial_{v_1}^2\Psi(0) =i.
  $$
This sets up an application of Theorem \ref{t:hor} where we integrate out $v_1$. If necessary, we shrink the neighborhood $K$ guaranteed by Theorem \ref{t:hor} in order to make the ideal generated by $\partial_{v_1}\Psi$ the
same as the ideal generated by $v_1+2iw_1$. One can do this
because
\begin{equation}\label{eq:der}
\partial_{v_1}\Psi=\tfrac{i}{2}(v_1+2iw_1)\bigg(\frac{1}{(1+u+w_2)^2+w_1^2-\tfrac{v_1^2}{4}-\tfrac{v_2^2}{4}+i\big(v_2(1+u+w_2)-v_1w_1\big)}+1\bigg),
\end{equation}
and the last term in the product vanishes only if
\begin{align*}
|v| &= 2\sqrt{(1+u+w_2)^2+w_1^2+1}.
\end{align*}
Since they are the same ideal, it is
easy to see that $h^0(v_2,u,w_1,w_2)=h(-2iw_1,v_2,u,w_1,w_2)$
where $h$ is either $f,\Psi$ or $\partial_{v_1}^2\Psi$.
The amplitude function
$f$ in \eqref{eq:amp} is not in $C_0^{\infty}(K)$, but it is a
Schwartz function in $(v_1,v_2)$ and the theorem still applies
by adding and subtracting a bump function whose support lies $K$. The remainder term is handled by Theorem 7.7.1 because by Lemma \ref{l:pos} and \eqref{eq:der} we have that $\Im \Psi>0$ or $\partial_{v_1}\Psi \neq 0$ outside of the support. Now applying
Theorem \ref{t:hor} to the $v_1$ integral of \eqref{eq:v1}, we
compute
\begin{align}
&\frac{1}{N^{5/3}}W_{\varphi_N^{\gamma_0},\varphi_N^{\gamma_0}}(x_{u,w_1,w_2}, \xi_{u,w_1,w_2})=
                                                          \frac{N^{5/6}}{4\pi^3\sqrt{2\pi}E^{2}}
                                                  e^{1-
                                                          \frac{1}{E}
                                                          |x_{u,w_1,w_2}|^2}
                                                          \big(1+O(\tfrac{1}{N})\big)
  \nonumber\\
                                                        &\qquad  \qquad \qquad \qquad \qquad \qquad \qquad \qquad     \qquad   \qquad                                                      
                                                          \int_{\mathbb{R}^2}
                                             f(v_1,v_2,u,w_1,w_2)e^{iN \Psi(v_1,v_2,u,w_1,w_2)}
                                                          dv \nonumber\\
  &=
                                                          \frac{N^{5/6}}{4\pi^3\sqrt{2\pi}E^{2}}
                                                  e^{1-
                                                          \frac{1}{E}
                                                          |x_{u,w_1,w_2}|^2}
                                                          \big(1+O(\tfrac{1}{N})\big)
  \nonumber\\
                                                        &\qquad  \qquad   \qquad                                                          
                                                          \bigg(\frac{\sqrt{2\pi}}{\sqrt{N}}\int_{\mathbb{R}}
                                             \frac{f^0(v_2,u,w_1,w_2)}{\sqrt{-i (\partial_{v_1}^2\Psi)^0(v_2,u,w_1,w_2)}}e^{iN \Psi^0(v_2,u,w_1,w_2)}dv_2+O(N^{-3/2})\bigg)
                                                          \nonumber\\
  &=
                                                          \frac{N^{5/6}}{4\pi^3\sqrt{2\pi}E^{2}}
                                                  e^{1-
                                                          \frac{1}{E}
                                                          |x_{u,w_1,w_2}|^2}
                                                          \big(1+O(\tfrac{1}{N})\big)
  \nonumber\\
                                                        &\qquad  \qquad   \qquad                                                          
                                                          \bigg(\frac{\sqrt{2\pi}}{\sqrt{N}}\int_{\mathbb{R}+2iw_2}
                                             \frac{f^0(v_2,u,w_1,w_2)}{\sqrt{-i (\partial_{v_1}^2\Psi)^0(v_2,u,w_1,w_2)}}e^{iN \Psi^0(v_2,u,w_1,w_2)}dv_2+O(N^{-3/2})\bigg)
                                                          \nonumber\\
  &=
                                                          \frac{N^{1/3}}{4\pi^3E^{2}}
                                                  e^{(N+1)(-u^2-2u-2w_1^2-2w_2^2)}
                                                          \big(1+O(\tfrac{1}{N})\big)\int_{\mathbb{R}}g(v_2,u)e^{iN \Phi(v_2,u)}dv_2
  \nonumber\\
                                                        &\qquad  \qquad   \qquad                                                          
                                                          +e^{1-
                                                          \frac{1}{E}
                                                          |x_{u,w_1,w_2}|^2}O(N^{-2/3})
                                                           \label{eq:v2}
\end{align}
where the penultimate equality follows from a change of contour, which we can do because of the decay of the integrand at infinity. For the last equality, we define
\begin{align}
g(v_2,u) \coloneqq &\ \frac{\sqrt{2}e^{-\frac{1}{4}v_2^2-iv_2(1+u)}}{\sqrt{\frac{1}{(1+u+i\frac{v_2}{2})^2}+1}}\label{eq:g}\\
\Phi(v_2,u) \coloneqq &\ -i\log\big[(1+u)^2+i
                            v_2(1+u)-\tfrac{v_2^2}{4}\big]+\tfrac{i}{4}v_2^2-v_2(1+u)\nonumber \\
  =&\ -i\log\big[(1+u+i\tfrac{v_2}{2})^2\big]+\tfrac{i}{4}v_2^2-v_2(1+u).  \nonumber                         
\end{align}
Note that if we view $\Phi$ as a function on 
$\mathbb{C}^2$, we see that it is holomorphic in a neighborhood
of $(0,0)$. Calculating derivatives of $\Phi$, we see
\begin{align}
\partial_{v_2}\Phi &=
                     \frac{1}{1+u+i\tfrac{v_2}{2}}+\tfrac{i}{2}v_2-(1+u) =-
    \frac{\tfrac{v_2^2}{4}+u(u+2)}{1+u+i\tfrac{v_2}{2}}\label{eq:dphi}\\
  \partial_{v_2}^2\Phi &=\frac{i}{2}\bigg(1-\frac{1}{(
                         1+u+i\tfrac{v_2}{2})^2}
                         \bigg)\label{eq:ddphi}\\
  \partial_{v_2}^3\Phi &=-\frac{1}{2}\frac{1}{( 1+u+i\tfrac{v_2}{2})^3 }\label{eq:dddphi}.
\end{align}
From \eqref{eq:dphi}, we see by the quadratic formula that
$\partial_{v_2}\Phi=0$ if and only if
$$v_2=v_{\pm}(u) \coloneqq \pm 2i\sqrt{u(u+2)}. $$
Additionally, note that $\partial_{v_2}\Phi(0,0)=\partial_{v_2}^2\Phi(0,0)=0$ and $\partial_{v_2}^3\Phi(0,0)=-1/2$, so we can apply Theorem \ref{t:CFU} (which follows from Theorem 2 of \cite{CFU56}) to see
\begin{align*}
e^{-N
  b(u)}\int_{\mathbb{R}}g(v_2,u)e^{iN \Phi(v_2,u)}dv_2&\sim
\tfrac{1}{N^{1/3}}\operatorname{Ai}(a(u)N^{2/3})\sum_{s=0}^{\infty}\mu_{0s}(u)N^{-s}
\\
  &\qquad \qquad \qquad +\tfrac{1}{N^{2/3}}\operatorname{Ai}'(a(u)N^{2/3})\sum_{s=0}^{\infty}\mu_{1s}(u)N^{-s}
\end{align*}
where $a$ and $b$ are functions satisfying
\begin{equation}\label{eq:T}
i\Phi(v_2,u)=\tfrac{1}{3}T(v_2,u)^3-a(u)T(v_2,u)+b(u)
\end{equation}
where $T$ is a function holomorphic in $v_2$ and $u$.  From \cite{CFU56}, the functions $a$ and
$b$ can be computed explicitly by
\begin{align}
  b(u)&=\tfrac{i}{2}\big[\Phi(v_{+}(u),u)+\Phi(v_{-}(u),u) \big]=u(u+2)\nonumber \\
  a(u)^{3/2}&=\tfrac{3i}{4}\big[\Phi(v_{+}(u),u)-\Phi(v_{-}(u),u)\big]\nonumber \\
  &=\tfrac{3}{2}\Big(2(u+1)\sqrt{u(u+2)}+\log
    \big[1+u-\sqrt{u(u+2)} \big]-\log
    \big[1+u+\sqrt{u(u+2)}\big]\Big) \nonumber\\
  &=3(u+1)\sqrt{u(u+2)}-3\log
    \big[1+u+\sqrt{u(u+2)}\big]\nonumber\\
      &=\begin{cases}
      3(u+1)i\sqrt{|u|(u+2)}-3i \operatorname{arccos}(1+u)&\text{ if }u<0\\
      3(u+1)\sqrt{u(u+2)}-3\operatorname{arccosh}(1+u)  &\text{ if }u\geq 0
      \end{cases} \label{eq:a(u)}
\end{align}
where the last line follows from the identities
\begin{align*}
  \operatorname{arccos}(z) = -i \log (z+i\sqrt{1-z^2}),\quad \operatorname{arccosh}(z) &= \log(z+\sqrt{z^2-1})
\end{align*}
where the principal branches are used for the logarithm and
square root functions. The branch in \eqref{eq:a(u)} is chosen
so as to make
\begin{equation}\label{eq:close}
a(u) \approx 2^{5/3}u \quad \text{for small }u,
\end{equation}
that is, we use the principal branch of the cube root (see the
top of page 603 in \cite{CFU56}). Now using \eqref{eq:v2}, we have
 \begin{align*}
\frac{1}{N^{5/3}}W_{\varphi_N^{\gamma_0},\varphi_N^{\gamma_0}}(x_{u,w_1,w_2}, \xi_{u,w_1,w_2})&\sim
                                                          \frac{\big(1+O(\tfrac{1}{N})\big)}{4\pi^3E^{2}}
                                                  e^{-2(N+1)(w_1^2+w_2^2)}
                                                         \bigg(\operatorname{Ai}(a(u)N^{2/3})\sum_{s=0}^{\infty}\mu_{0s}(u)N^{-s} 
  \nonumber\\
                                                        &\qquad  \qquad   \qquad                                                          
                                                          +\tfrac{1}{N^{2/3}}\operatorname{Ai}'(a(u)N^{1/3})\sum_{s=0}^{\infty}\mu_{1s}(u)N^{-s}\bigg)\\
 &\qquad  \qquad   \qquad                                                            +e^{1-
                                                          \frac{1}{E}
                                                          |x_{u,w_1,w_2}|^2}O(N^{-2/3}).
 \end{align*}
 Now we need to compute $\mu_{00}$. At the end of page 601 of
 \cite{CFU56}, we see
 \begin{equation}\label{eq:mu}
\frac{\mu_{00}(u)}{2\pi i}=\frac{1}{2}\Big(g(v_{+},u)\big(\partial_{v_2}T(v_+,u)\big)^{-1}+g(v_{-},u)\big(\partial_{v_2}T(v_-,u)\big)^{-1} \Big),
\end{equation}
so it suffices to compute $\partial_{v_2}T(v_{\pm},u)$. For
this, by Theorem 2 of \cite{CFU56}, we see
\begin{equation}\label{eq:relate}
T(v_{\pm},u)=\mp a(u)^{1/2}.
\end{equation}
Note that \eqref{eq:relate} cannot be $\pm a(u)^{1/2}$ for otherwise this would switch the sign of \eqref{eq:sr}, contradicting \eqref{eq:contra}.
Now we differentiate \eqref{eq:T} twice in $v_2$ to obtain
\begin{equation}\label{eq:nothing}
i\partial_{v_2}^2\Phi = \partial_{v_2}^2T\big(T^2-a(u) \big)+ \big( \partial_{v_2}T\big)^22T.
\end{equation}
Evaluating at $v_2=v_{\pm}$, we see
\begin{align*}
i\partial_{v_2}^2\Phi(v_{\pm},u) &=
                                 \partial_{v_2}^2T(v_{\pm},u)\big(T^2(v_{\pm},u)-a(u)
                                 \big)+ \big(
                                 \partial_{v_2}T(v_{\pm},u)\big)^22T(v_{\pm},u)
                                \\
  &\overset{\mathclap{\eqref{eq:relate}}}{=} \mp 2 \big(
                                 \partial_{v_2}T(v_{\pm},u)\big)^2a(u)^{1/2}.
\end{align*}
Using \eqref{eq:ddphi}, we see
\begin{equation}\label{eq:sr}
\big(\partial_{v_2}T(v_{\pm},u)\big)^2 = \frac{\frac{1}{(1+u+iv_{\pm}/2)^2}-1}{\mp 4 a(u)^{1/2}}=-\frac{\pm u(u+2) + (1+u)\sqrt{u(u+2)}}{ 2a(u)^{1/2}}.
\end{equation}
Since $\partial_{v_2}T$ is holomorphic, it suffices to find the
sign of the square root. For this, note
$\big(\partial_{v_2}T(0,0)\big)^2=-2^{-4/3}$ by setting $u=0$ in \eqref{eq:sr}. On the other hand,
by setting $u=0$ in \eqref{eq:T}, we have
\begin{equation*}
2\log(1+i\tfrac{v_2}{2})-\tfrac{i}{4}v_2^2+v_2=\tfrac{1}{3}T(v_2,0)^3.
\end{equation*}
Dividing both sides by $v_2^3$ and letting $v_2\to 0$, we see
\begin{equation}\label{eq:contra}
-\tfrac{i}{4}=\big(\partial_{v_2}T(0,0)\big)^3,
\end{equation}
which implies $\partial_{v_2}T(0,0)=2^{-2/3}i$, so the square root in \eqref{eq:sr} is the principal square root. In particular, we see $\mu_{00}(0)=2^{5/3}\pi$ (for an explicit formula for $\mu_{00}(u)$, see Section \ref{s:mu}). Lastly, we use Lemma \ref{lem:rotation} to rotate our initial point $(1,i)=((1,0),(0,1))$ around the orbit giving the desired
 result for $\gamma_0$. We then use \eqref{eq:meta2} to obtain the result for general $\gamma$. Note that the statement of the Theorem is independent of how we extend $(x,\xi)$ to an orthogonal basis (each vector with norm $\sqrt{2E}$) of $N_{(x,\xi)}\gamma$ with the vectors $v_1,v_2$ since the approximation only depends on the coefficients of these vectors, concluding the proof.
\section{Relation with the Eigenspace Projector}\label{s:p}
Suppose $\Pi_N$ is the eigenspace projector to the $\hbar (N+1)$-eigenspace $\mathcal{H}_N$ of \eqref{eq:harmos}. The orbital coherent states have the following reproducing formula.
\begin{prop}\label{p:proj} If $m$ is the Haar measure on $\operatorname{SU}(2)$ and $\mu$ is the metaplectic representation on $L^2(\mathbb{R}^2)$, then we have
\begin{equation}\label{eq:pro1}
\Pi_N
  =(N+1)\int_{\operatorname{SU}(2)}
  \big(\mu(U)\varphi_N^{\gamma_0} \big)\otimes \big(\mu(U)\varphi_N^{\gamma_0}\big)^{\dagger}\, dm(U).
\end{equation}
In particular, if $\nu$ is the normalized Fubini-Study measure on $\mathbb{C} \text{P}^1$, we have
\begin{equation}\label{eq:pro2}
\Pi_N=(N+1)\int_{\mathbb{C}\text{P}^1} \varphi_N^{\gamma}\otimes (\varphi_N^{\gamma})^{\dagger}\, d\nu(\gamma).
\end{equation}
\end{prop}
\begin{proof}
We show \eqref{eq:pro2} as the proof of \eqref{eq:pro1} is the same. We do this by showing that the integral on the right-hand side of \eqref{eq:pro2} is a constant
times the identity using Schur's lemma (see Schur's lemma I on page 262 of \cite{F89}). Let $\mu :\operatorname{SU}(2) \to
\mathcal{U}(\mathcal{H}_N)$ be the subrepresentation of the metaplectic representation restricted to $\operatorname{SU}(2)$. It is a standard fact that this representation is irreducible since it is rotating homogeneous polynomials in two complex variables with constant degree in Fock space. Now for all $g \in \operatorname{SU}(2)$ and $f \in
\mathcal{H}_N,$ we have
\begin{align*}
\mu(g)\cdot \int_{\mathbb{C}\text{P}^1}\langle
  \varphi_N^{\gamma},f \rangle
  \varphi_N^{\gamma} d\nu(\gamma)&=
                                                 \int_{\mathbb{C}\text{P}^1}\langle
                                                 \varphi_N^{\gamma},f
                                                 \rangle
                                                 \mu(g)[
                                                 \varphi_N^{\gamma}]
                                                 d\nu(\gamma)\\
 &=
                                                 \int_{\mathbb{C}\text{P}^1}\langle
                                                 \varphi_N^{g^{-1}\cdot \gamma},f
                                                 \rangle
                                                 \mu(g)[
                                                 \varphi_N^{g^{-1}\cdot \gamma}]
                                                 d\nu(\gamma)\\
    &=
                                                 \int_{\mathbb{C}\text{P}^1}\langle
                                                 \mu(g^{-1})[\varphi_N^{\gamma}],f
                                                 \rangle                                                 
                                                 \varphi_N^{ \gamma}
                                                 d\nu(\gamma)\\
  &=
                                                 \int_{\mathbb{C}\text{P}^1}\langle
                                                 \varphi_N^{\gamma},\mu(g)[f]
                                                 \rangle                                                
                                                 \varphi_N^{\gamma}
                                                 d\nu(\gamma)
\end{align*}
where the first equality we use the Fubini-Tonelli theorem as $\mu(g)$ is
given by integral formulas starting on page 191 of
\cite{F89}. For the second line, we use the fact that the Fubini-Study measure is $\operatorname{SU}(2)$-invariant. The third line follows because 
$$\mu(g^{-1})[\varphi_N^{\gamma}]=\varphi_N^{g^{-1} \cdot
  \gamma}$$
up to a phase, which follows from \eqref{eq:genphi}. Schur's lemma gives
$$\Pi_N=c_N\int_{\mathbb{C}\text{P}^1}
\varphi_N^{\gamma} \otimes (\varphi_N^{\gamma})^{\dagger}\, d\nu(\gamma) $$
for some constant $c_N$, which can be computed as $N+1$ by taking the trace of both sides.
\end{proof}
Taking Wigner distributions of both sides of \eqref{eq:pro2} and
using the Fubini-Tonelli theorem, we see that
\begin{equation}\label{eq:wigcomp}
W_{\Pi_N}(x,\xi) \coloneqq  \int_{\mathbb{R}^2}\Pi_N (x+\tfrac{v}{2},x-\tfrac{v}{2})e^{-\frac{i}{\hbar
                                                             }\langle
                                         v ,
                                                             \xi
                                                             \rangle}\
                                                             \frac{dv}{4\pi^2\hbar^2}=(N+1)\int_{\mathbb{C}\text{P}^1} W_{ \varphi_N^{\gamma},\varphi_N^{\gamma}}(x,\xi)\, d\nu(\gamma).
\end{equation}
The Wigner distributions $W_{\Pi_N}$ were studied extensively in \cite{HZ20}. Theorem 1.4 of \cite{HZ20} proves interface Airy asymptotics for $W_{\Pi_N}$, and a consequence of their work is the following:
\begin{theo} For $u \in \mathbb{R}$ and $(x, \xi) \in \Sigma_E$, we have
$$\lim_{N \to \infty}\frac{1}{N^{5/3}}W_{\Pi_N}\big[(1+\tfrac{u}{(2N)^{\frac{2}{3}}})(x,\xi)\big]= \frac{1}{2^{4/3}\pi^2E^2}\operatorname{Ai}(2u).$$
\end{theo}
This compares nicely with Theorem \ref{t:scale}. In fact, if we set $w_1=w_2=0$ in Theorem \ref{t:scale}, the same conclusion is arrived up to a factor of $2$. Theorem \ref{t:scale} along with \eqref{eq:wigcomp} shows that the interface Airy asymptotics of $W_{\Pi_N}(x,\xi)$ in \cite{HZ20} are coming from the Airy asymptotics of $W_{\varphi_N^{\gamma},\varphi_N^{\gamma}}(x,\xi)$ in \ref{t:scale} for $\gamma$ in a small neighborhood of the orbit containing $(x,\xi)$.
\section{Generalizations to $d$ dimensions}\label{s:generalize}
The two dimensional coherent state that concentrates on
$\gamma_0$ in configuration space is
\begin{equation*}\label{eq:coherent}
 \varphi_N^{\gamma_0}(x_1,x_2)=\frac{(N+1)^{\frac{N+1}{2}}}{E^{\frac{N+1}{2}}\sqrt{\pi
    N!}}(x_1+ix_2)^Ne^{-\frac{(N+1)}{2E}|x|^2}=\frac{1}{\hbar^{\frac{N+1}{2}}\sqrt{\pi
    N!}}(x_1+ix_2)^Ne^{-\frac{1}{2\hbar}|x|^2}.
\end{equation*}
where $E=\hbar(N+1)$.
We can extend this coherent state
to a coherent state on dimension $d$ by multiplying by one-dimensional ground
states:
\begin{align*}
\varphi_{N,d}^{\gamma_0}(x_1,\ldots,x_d)\coloneqq & \ \varphi_N^{\gamma_0}(x_1,x_2)\frac{1}{(\hbar \pi)^{\frac{d-2}{4}}}e^{-\tfrac{1}{2\hbar}x_3^2-\cdots -\tfrac{1}{2\hbar}x_d^2}\\
  =&\ \frac{1}{(\hbar \pi)^{\frac{d-2}{4}}\hbar^{\frac{N+1}{2}}\sqrt{\pi
    N!}}(x_1+ix_2)^Ne^{-\frac{1}{2\hbar}|x|^2}.\\
  =&\ \frac{1}{\hbar^{N+\frac{d}{2}} \pi^{\frac{d}{4}}\sqrt{
    N!}}(x_1+ix_2)^Ne^{-\frac{1}{2\hbar}|x|^2}.
\end{align*}
Here we view $\gamma_0$ on the left hand side as $\gamma_0$
embedded in $d$-dimensional configuration space $\mathbb{R}^d$.
Recall recall that the Wigner distribution $W_{\varphi_N^{\gamma},\varphi_N^{\gamma}}:T^*\mathbb{R}^d \to \mathbb{R}$ of $\varphi_N^{\gamma}$ is
defined by
\begin{align*}
W_{\varphi_N^{\gamma},\varphi_N^{\gamma}}(x,\xi) \coloneqq &\
 \frac{1}{(2\pi)^d\hbar^d}\int_{\mathbb{R}^d}\varphi_N^{\gamma}(x+\tfrac{v}{2})\overline{\varphi_N^{\gamma}(x-\tfrac{v}{2})}e^{-\frac{i}{\hbar
                                                             }\langle
                                         v ,
                                                             \xi
                                                             \rangle}\
                                                             dv.
\end{align*}
Also recall that the one-dimensional ground state $g(x)=(\hbar
\pi)^{-\frac{1}{4}}e^{-\frac{1}{2\hbar}x^2}$ has Wigner
distribution
$$W_{g, g}(x,\xi)=\frac{1}{\pi \hbar}e^{-\frac{1}{\hbar}(x^2+\xi^2)}.$$
We then have
\begin{align}
W_{\varphi_{N,d}^{\gamma_0},
  \varphi_{N,d}^{\gamma_0}}(x,\xi)&=\frac{1}{(\pi\hbar)^{d-2}}e^{-\frac{1}{\hbar}(x_3^2+\xi_3^2+\cdots+x_d^2+\xi_d^2)}W_{\varphi_{N}^{\gamma_0},
                                    \varphi_{N}^{\gamma_0}}\big((x_1,x_2),(\xi_1,\xi_2)\big)\nonumber \\
  &=\bigg(\frac{N+\frac{d}{2}}{\pi E}\bigg)^{d-2}e^{-\frac{N+\frac{d}{2}}{E}(x_3^2+\xi_3^2+\cdots+x_d^2+\xi_d^2)}W_{\varphi_{N}^{\gamma_0},
                                    \varphi_{N}^{\gamma_0}}\big((x_1,x_2),(\xi_1,\xi_2)\big),\label{eq:gen}
\end{align}
where $E=\hbar (N+\frac{d}{2})$.
The non-constant Hamiltonian orbits of the $d$-dimensional
isotropic harmonic oscillator are still great circles lying in
two-dimensional planes, and we still have metaplectic covariance
\begin{equation}\label{eq:meta3}
W_{\mu(U)f, \mu(U)f}(x,\xi)=W_{f, f}\big(U^*(x,\xi)\big)
\end{equation}
where $f \in L^2(\mathbb{R}^d)$, $U \in \operatorname{SU}(d)$, and $\mu$ is the metaplectic representation on $L^2(\mathbb{R}^d)$. The non-constant Hamiltonian orbits can be similarly identified by $\mathbb{C}\text{P}^{d-1}$. Since $\operatorname{SU}(d)$ acts transitively on $\mathbb{C}\text{P}^{d-1}$ by rotation, for any non-constant Hamiltonian orbit $\gamma$ there exists a $U \in \operatorname{SU}(d)$ such that $U \cdot \gamma_0=\gamma$. We define (up to a phase) the coherent state centered at the Hamiltonian orbit $\gamma$ by
\begin{equation*}
\varphi_{N,d}^{\gamma} \coloneqq \mu(U)\varphi_{N,d}^{\gamma_0}.
\end{equation*}
Using \eqref{eq:meta3}, we have
\begin{equation}\label{eq:meta4}
W_{\varphi_{N,d}^{\gamma}, \varphi_{N,d}^{\gamma}}(x,\xi)=W_{\varphi_{N,d}^{\gamma_0},\varphi_{N,d}^{\gamma_0}}\big(U^*(x,\xi)\big).
\end{equation}
With \eqref{eq:gen} along with \eqref{eq:meta4}, we can prove
all of our results in dimension $d$. We state them below.
\begin{prop}\label{t:deltad} Let $a:T^*\mathbb{R}^d \to \mathbb{R}$ be a smooth function with exponential decay. Then in the limit $N\to \infty$ where $E=\hbar (N+d/2)$, we have
\begin{equation}\label{eq:me}
\int_{T^*\mathbb{R}^d}a(x,\xi)W_{\varphi_{N,d}^{\gamma},\varphi_{N,d}^{\gamma}}(x,\xi)\ dxd\xi \xrightarrow{N \to \infty} \frac{1}{2\pi}\int_0^{2\pi}a\big(\gamma(t)\big)\ dt.
\end{equation}
\end{prop}
\begin{theo}\label{t:main24}
As $N \to \infty$, we have the following pointwise asymptotics:
\begin{itemize}
\item[(1)] If $(x,\xi)$ lies on the orbit $\gamma$, then we have
$$\frac{1}{N^{d-1/3}}W_{\varphi_{N,d}^{\gamma},\varphi_{N,d}^{\gamma}}(x,\xi)=\frac{1}{2^{1/3}\pi^dE^d}\operatorname{Ai}(0)+O(N^{-2/3}).$$
where $\frac{1}{2^{1/3}\pi^dE^d}\operatorname{Ai}(0)=\frac{\Gamma(1/3)12^{1/3}}{4\pi^{d+1}E^d\sqrt{3}}$
\item[(2)] If $(x,\xi)$ lies on the orbit $\gamma$, the for $0<t< 1$ we have 
$$\frac{1}{N^{d-1/3}}W_{\varphi_{N,d}^{\gamma},\varphi_{N,d}^{\gamma}}\big(t(x,\xi)\big)=O(N^{-1/6}).$$
\item[(3)] If $(x,\xi)\neq 0$ is not in the above case, we have
$$W_{\varphi_{N,d}^{\gamma},\varphi_{N,d}^{\gamma}}(x,\xi) =
O(N^{-\infty}).$$
\item[(4)] If $(x,\xi)=0$, then the Wigner distribution is explicitly $$W_{\varphi_{N,d}^{\gamma},\varphi_{N,d}^{\gamma}}(x,\xi)=\frac{(-1)^N}{\pi^2E^d}(N+1)^d.$$
\end{itemize}
\end{theo}
\begin{theo}\label{t:scale2} Suppose
$(x,\xi)$ lies on $\gamma$, and let $(x,\xi),v_1,\ldots,v_{2(d-1)}$ be an orthogonal basis (with each vector of norm $\sqrt{2E}$) for the normal space, $N_{(x,\xi)}\gamma$, of $\gamma$ at $(x,\xi)$. For $(u,w_1,\ldots,w_{2(d-1)}) \in \mathbb{R}^{2d-1}$
 in a sufficiently small neighborhood of $0$, we have the uniform asymptotic expansion
 \begin{align*}
   \frac{1}{N^{d-1/3}}W_{\varphi_{N,d}^{\gamma},\varphi_{N,d}^{\gamma}}&\big[(1+u)(x,\xi)+w_1v_1+\cdots+w_{2(d-1)}v_{2(d-1)}\big]\\
                                                                                  &=
                                                          \frac{ e^{-2(N+1)(w_1^2+\cdots+w_{2(d-1)}^2)}}{4\pi^{d+1}E^{d}}                                                 
                                                         \mu_{00}(u)\operatorname{Ai}(a(u)N^{2/3})+\operatorname{Ai}'(a(u)N^{2/3})\cdot O(N^{-2/3}) 
 \end{align*}
 where
 $$a(u)^{3/2}= \begin{cases}
      3(u+1)i\sqrt{|u|(u+2)}-3i \operatorname{arccos}(1+u)&\text{ if }u<0\\
      3(u+1)\sqrt{u(u+2)}-3\operatorname{arccosh}(1+u)  &\text{ if }u\geq 0
      \end{cases}$$
and $\mu_{00}(u)$ is a smooth function (explicitly given in the appendix) with $\mu_{00}(0)=2^{5/3}\pi$.
In particular, if we rescale $u \to \frac{u}{(2N)^{\frac{2}{3}}}$, $w_1 \to \frac{w_1}{(2N)^{\frac{1}{2}}}$, and $w_2 \to \frac{w_2}{(2N)^{\frac{1}{2}}}$, we have the pointwise limit
$$\lim_{N \to \infty}\frac{1}{N^{d-1/3}}W_{\varphi_{N,d}^{\gamma},\varphi_{N,d}^{\gamma}}\big[(1+\tfrac{u}{(2N)^{\frac{2}{3}}})(x,\xi)+\tfrac{w_1}{(2N)^{\frac{1}{2}}}v_1+\cdots+\tfrac{w_{2(d-1)}}{(2N)^{\frac{1}{2}}}v_{2(d-1)}\big]= \frac{e^{-(w_1^2+\cdots+w_{2(d-1)}^2)}}{2^{1/3}\pi^dE^d}\operatorname{Ai}(2u).$$
\end{theo}
Lastly, the same proof of Proposition \ref{p:proj} works for dimension $d$, and we have similar comparisons to results in \cite{HZ20}.
\section{Appendix}
\subsection{Airy Function}\label{s:airy}
The Airy function $\operatorname{Ai}(x)$ on $\mathbb{R}$ is
defined as the inverse Fourier transform of $\xi \mapsto
e^{i\xi^3/3}$. This implies $\operatorname{Ai}\in
\mathcal{S}'(\mathbb{R})$, but if we view the integral as an
improper Riemann integral on $[-L,L]$ with $L \to \infty$, we
see that $\operatorname{Ai} \in C^{\infty}(\mathbb{R})$ by an
integration by parts argument. So we define
$$\operatorname{Ai}(x) \coloneqq \lim_{L \to \infty}\frac{1}{2\pi}\int_{-L}^L
\cos(\xi^3/3+\xi x)d\xi.$$
This function satisfies the differential equation
$$\operatorname{Ai}''(x)-x\operatorname{Ai}(x)=0, $$
and it can be extended to an entire function on $\mathbb{C}$ using this differential equation.
\subsection{Stationary Phase Methods}\label{s:phase}
For the pointwise asymptotics, we use a consequence of Theorem
7.7.5 of \cite{H03}, which we state below
\begin{theo}\label{t:horb}
Suppose $f,\Psi$ are $C^{\infty}$ complex valued functions on $\mathbb{R}^d$. Suppose that
$$\Im \Psi \geq 0, \Im \Psi (x_0)=0, \Psi'(x_0)=0, \det
\Psi''(x_0) \neq 0, $$
with $\Psi' \neq 0$ in a neighborhood of $x_0$. Then
$$\int_{\mathbb{R}^d} f(x)e^{i N
  \Psi(x)}dx=\Big(\frac{2\pi}{N}\Big)^{\frac{d}{2}}\det(-i
    \Psi''(x_0) )^{-\frac{1}{2}}f(x_0)e^{i N \Psi(x_0)}+O(N^{-\frac{d}{2}-1}) $$
provided that $\operatorname{supp}f$ is in a sufficiently small neighborhood of $x_0$.
\end{theo}
Also, for the scaling asymptotics we use Theorem 7.7.12 of \cite{H03}, which is the above theorem with an extra variable (see below).
\begin{theo}\label{t:hor} Let $\Psi(x,y)$ be a
complex valued $C^{\infty}$ function in a neighborhood of
$(0,0)$ in $\mathbb{R}^{1+d}$, satisfying
$$\Im \Psi(0,0)=0,\quad \Im \Psi \geq 0, \quad \partial_x
\Psi(0,0)=0,\quad \partial_x^2\Psi(0,0)\neq 0, $$
and let $f \in C_0^{\infty}(K)$ where $K$ is a small
neighborhood of $(0,0)$ in $\mathbb{R}^{1+d}$. Then
$$\int_{\mathbb{R}} f(x,y)e^{i N
  \Psi(x,y)}dx=\frac{\sqrt{2\pi}}{\sqrt{N}\sqrt{-i
    \big(\partial_x^2\Psi \big)^0(y)}}f^0(y)e^{i N \Psi^0(y)}+O(N^{-3/2}) $$
where for functions $g(x,y)$ the notation $g^0(y)$ stands for a function of $y$ only which is in the same residue class modulo the ideal generated by $\partial_x\Psi$.
\end{theo}
The next theorem we reproduce is a consequence of Theorem 2 of \cite{CFU56},
which we use at a critical point in the proof of the scaling
asymptotics.
\begin{theo}\label{t:CFU} Let $\Psi(z,w)$ and $f(z,w)\in \mathcal{S}(\mathbb{C}^2)$ be
complex valued functions holomorphic in both variables in a neighborhood of
$(0,0)$ in $\mathbb{C}^2$, where $\Psi$ satisfies
$\partial_z\Psi(0,0)=\partial_z^2\Psi(0,0)=0$ and $\partial_z^3\Psi(0,0)\neq
0$. Further suppose that $\Psi$ has only two nondegenerate critical points in $z$,  $\pm z(w)$, that depend on $w$ and differ by a sign that coalesce to $0$ as $w$ goes to $0$ (i.e. $z(0)=0$). Then uniformly in $w$, we have the asymptotic formula
\begin{align*}
e^{-i N
                 b(w)}\int_{\mathbb{R}} f(z,w)e^{i N
  \Psi(z,w)}dz &\sim N^{-1/3}\operatorname{Ai}\big(a(w)N^{2/3}\big)\sum_{k=0}^{\infty}u_{0k}(w)N^{-k}\\
  &\qquad \qquad+N^{-2/3}\operatorname{Ai}'\big(a(w)N^{2/3}\big)\sum_{k=0}^{\infty}u_{1k}(w)N^{-k},
\end{align*}
provided that $w$ is in a sufficiently small neighborhood of $0$. The functions $a,b,u_{jk}$ are holomorphic on this neighborhood. 
\end{theo}
This theorem is the holomorphic version of Theorem 7.7.18 of \cite{H03}. The functions $a$ and $b$ are related to $\Psi$ by the equation
$$i\Psi(z,w)=\tfrac{1}{3}T(z,w)^3-a(w)T(z,w)+b(w)$$
where $T$ is a holomorphic function in $z$ and $w$. For more
information, see \cite{CFU56}. The last theorem is a stationary
phase approximation for Bott-Morse functions, and we use it in
the weak limit proof.
\begin{theo}[Stationary Phase for Bott-Morse functions]\label{t:BM}
Let $U \subset \mathbb{R}^d$ be an open set, and let $f,\Psi \in
C^{\infty}(U)$ be such that $\Im \Psi \geq 0$ in $U$ and
$\operatorname{supp}f \subset U$. We define
$$C \coloneqq \{u \in U \mid \Im \Psi(u)=0, \Psi'(u)=0\}.$$
Assume that $C$ is a smooth, compact, and connected manifold of
$\mathbb{R}^d$ of dimension $k$ such that for all $u \in C$ the
Hessian, $\Psi''(u)$, of $\Psi$ is nondegenerate on the normal
space $N_uC$ to $C$ at $u$. Then we have
$$\int_{\mathbb{R}^d}f(x)e^{i N
  \Psi(x)}dx=\Big(\frac{2\pi}{N}\Big)^{\frac{d-k}{2}}e^{i
  \Psi(u_0)}\int_C\frac{f(u)}{\sqrt{\det(-i
    \Psi''|_{N_uC})}}d\mu_C(u)+O(N^{-\frac{d-k}{2}-1}) $$
where $u_0 \in C$ is arbitrary, $\mu_{C}$ denotes the natural Riemannian metric on $C$, and the denominator of the integrand is the product of the square roots of the eigenvalues of $-i
    \Psi''|_{N_uC}$, chosen with positive real part.
    \end{theo}
For the proof of this theorem, see Theorem 29 on page 134 of \cite{CR12}
\subsection{Formula for $\mu_{00}(u)$}\label{s:mu}
From \eqref{eq:mu} in Section \ref{s:scale}, we have
$$\frac{\mu_{00}(u)}{2\pi
  i}=\frac{1}{2}\Big(g(v_{+},u)\big(\partial_{v_2}T(v_+,u)\big)^{-1}+g(v_{-},u)\big(\partial_{v_2}T(v_-,u)\big)^{-1}
\Big),$$
where $g$ is defined at \eqref{eq:g} by
$$g(v,u) \coloneqq \frac{\sqrt{2}e^{-\frac{1}{4}v_2^2-iv_2(1+u)}}{\sqrt{\frac{1}{(1+u+i\frac{v_2}{2})^2}+1}},$$
and \eqref{eq:sr} gives
$$\partial_{v_2}T(v_{\pm},u)= \sqrt{\frac{\frac{1}{(1+u+iv_{\pm}/2)^2}-1}{\mp 4 a(u)^{1/2}}}=i\sqrt{\frac{\pm u(u+2) + (1+u)\sqrt{u(u+2)}}{ 2a(u)^{1/2}}}.$$
Combing all of these equations, we have
\begin{align*}\label{eq:bigmu}
  \mu_{00}(u)&=2\pi|a(u)|^{1/4}\frac{e^{u(u+2)}}{\sqrt{2u(u+1)(u+2)}}\bigg|e^{2(u+1)\sqrt{u(u+2)}}\sqrt{\frac{-u(u+2)+(u+1)\sqrt{u(u+2)}}{1+u+\sqrt{u(u+2)}}}\\
             &\qquad+e^{-2(u+1)\sqrt{u(u+2)}}\sqrt{\frac{u(u+2)+(u+1)\sqrt{u(u+2)}}{1+u-\sqrt{u(u+2)}}} \bigg|.
\end{align*}

\end{document}